\documentclass[article,onecolumn]{IEEEtran}
\usepackage{graphicx}
\usepackage{amssymb,amsmath}
\usepackage{amsthm}
\usepackage{amsfonts}
\usepackage{float} 
\usepackage{srcltx}
\usepackage{mathrsfs} 
\usepackage{multirow}
\usepackage{bbm}

\usepackage{graphicx}
\usepackage{epsfig}
\usepackage{psfrag}

\usepackage{setspace}

\usepackage[]{units} 
\usepackage{url} 
\usepackage[dvips]{color}
\usepackage{verbatim} 

\newcommand{\svv}[1]{\mathbf{#1}}
\DeclareRobustCommand{\Nt}[1][Nt]{\ensuremath {N_t}}
\DeclareRobustCommand{\alNt}[1][Nt]{\alpha(K)}
\DeclareRobustCommand{\aNorm}[1][aNorm]{\ensuremath {\|{\bf a}\|}}

\newcommand{\SNR}{\text{$\mathsf{SNR}$}}

\newtheorem{lemma}{Lemma}

\newtheorem{theorem}{Theorem}
\newtheorem{definition}{Definition}

\newcommand{\diag}{\mathop{\mathrm{diag}}}

\DeclareRobustCommand{\prob}[1][{\rm Pr}]{\ensuremath {{#1}}}
\DeclareRobustCommand{\genGam}[1][\beta]{\ensuremath {{#1}}}

\begin{document}
\allowdisplaybreaks
%
\title{Achievability Performance Bounds for Integer-Forcing Source Coding}

\author{Elad Domanovitz and Uri Erez
\thanks{The work of E. Domanovitz and U. Erez was supported in part by the Israel Science Foundation under Grant No. 1956/17 and
 by the Heron consortium via the Israel Ministry of Economy
and Industry.}
\thanks{The material in this paper was presented in part at the 2017 IEEE
Information Theory Workshop, Kaohsiung.}
\thanks{E. Domanovitz and U. Erez are with the Department of Electrical Engineering -- Systems, Tel Aviv University, Tel Aviv, Israel (email: domanovi,uri@eng.tau.ac.il).}
}
\maketitle



\begin{abstract}
Integer-forcing source coding has been proposed as  a low-complexity method for compression of distributed correlated Gaussian sources. In this scheme, each encoder
quantizes its observation using the same fine lattice and reduces the result modulo a coarse lattice. Rather than directly recovering the individual
quantized signals, the decoder first recovers a full-rank set of
judiciously chosen integer linear combinations of the quantized
signals, and then inverts it.
It has been observed that the method works very well for ``most" but not all source covariance matrices. The present work quantifies the measure of bad covariance matrices by studying the probability that integer-forcing source coding fails as a function of the allocated rate, where the probability is with
respect to a random orthonormal transformation that is applied to the sources prior to quantization. For the important case where the signals to be compressed correspond to the antenna inputs of relays in an i.i.d. Rayleigh fading environment, this orthonormal transformation can be viewed as being performed by nature.
The scheme is also studied in the context of a non-distributed system. Here, the goal is to arrive at a universal, yet practical, compression method using equal-rate quantizers with provable performance guarantees. The scheme is universal in the sense
that the covariance matrix need only be learned at the decoder
but not at the encoder. The goal is accomplished by replacing the random orthonormal transformation by  transformations corresponding to  number-theoretic space-time codes.

\end{abstract}


%
\IEEEpeerreviewmaketitle

\section{Introduction}

Integer-forcing (IF) source coding, proposed in \cite{Integer-ForcingSourceCoding:OrdentlichErez2017}, is a
scheme for distributed lossy compression of correlated Gaussian sources under a minimum mean squared error distortion measure.
Similar to its channel coding counterpart, in this scheme, all
encoders use the same nested lattice codebook.
Each encoder quantizes its observation using the fine lattice as a quantizer and reduces the result modulo the coarse lattice, which plays the role of binning. Rather than directly recovering the individual quantized signals, the decoder first recovers a full-rank set of judiciously chosen integer linear combinations of the quantized signals, and then inverts it. An appealing feature of integer-forcing source coding, not shared by previously proposed practical methods (e.g., Wyner-Ziv coding) for the distributed source coding problem, is its inherent symmetry, supporting equal distortion and  quantization rates.
A potential application of IF source coding is to distributed compression of signals received at several relays as suggested in  \cite{Integer-ForcingSourceCoding:OrdentlichErez2017} and further explored in \cite{ordonez2016integer}.



Similar to IF channel coding, IF source
coding works well for ``most" but not all
Gaussian vector sources.
Following the approach of \cite{OutageBehaviorOfRandomlyPrecoded:Domanovitz2015}, in the present work
we quantify the measure of bad source covariance matrices by considering  a randomized version of IF source coding where a random orthogonal transformation  is applied to the sources prior to quantization. Specifically, we derive bounds on the worst-case (with respect to a compound class of Gaussian sources)  outage probability when the precoding matrices are drawn from the circular real ensemble (CRE); that is, they are uniformly distributed with respect to the Haar measure. While in general such a transformation implies joint processing at the encoders, we note that in some natural scenarios,  including that of distributed compression of signals received at relays in an i.i.d. Rayleigh fading environment, the random transformation is actually performed by nature.\footnote{This follows since the left and right singular vector matrices of an i.i.d. Gaussian matrix $\svv{K}$ are equal to the eigenvector matrices of the  Wishart ensembles $\svv{K}\svv{K}^{T}$ and $\svv{K}^T\svv{K}$, respectively. The latter are known to be uniformly (Haar) distributed. See, e.g., Chapter~4.6 in \cite{edelman2005random}.}
In fact, it was already empirically observed in  \cite{Integer-ForcingSourceCoding:OrdentlichErez2017}
that IF source  coding performs very well in the latter scenario.

We then consider the performance of IF source coding when  used
in conjunction with judiciously chosen (deterministic) unitary precoding.
Such precoding requires joint processing of the different sources prior to quantization and is thus precluded in a distributed setting. Nonetheless, the scheme has an important advantage with respect to traditional centralized source coding of correlated sources. 
Whereas the traditional approach requires utilizing the statistical characterization of the source at the encoder side, e.g., via the application of appropriate transform coding, 
the considered compression scheme, in contrast, applies a \emph{universal} transformation, i.e., a transformation that is independent of the source statistics. Furthermore, no bit loading is needed. That is, the components of the output  of the transformation are all quantized at the same rate and hence the operation of the encoder does not depend on the source statistics. 
This characteristic may be advantageous in scenarios where the complexity of the encoder should be kept to a minimum whereas larger computational resources are available in the reconstruction stage where of course knowledge of the statistics of the source has to be obtained and utilized in some manner. We refer the reader to \cite{romanov2019blind} which describes a practical method by which the decoder may estimate the statistics of the source from the modulo-reduced quantized samples, whereas the encoder may remain oblivious to the source statistics.

For such a centralized application of IF source coding, we are able to derive stronger performance guarantees than those we derive for the distributed setting. Namely, in the centralized setting, we show that by employing transformations derived from algebraic number-theoretic constructions, successful source reconstruction can be guaranteed for any Gaussian source in the considered compound class. Thus, we need not allow for outage events. For  general Gaussian sources, such guarantees require performing precoding jointly over space and time. In contrast, for the case of parallel (independent) Gaussian sources with different variances, space-only precoding suffices. This distinction closely mirrors the well-known results concerning transformations for achieving maximal diversity for communication over fading channels, as described in \cite{oggier2004algebraic}.

The rest of this paper is organized as follows.
Section~\ref{sec:formulation} formulates the
problem of distributed compression of Gaussian sources in a compound vector source setting and provides some relevant background on IF source coding.
Section~\ref{sec:cre} describes randomly precoded
IF source coding and its empirical performance.
Section~\ref{sec:upper} derives upper bounds on the probability of failure of randomly-precoded
IF as a function of the excess rate.
In Section~\ref{sec:additive}, deterministic linear precoding is considered. A bound on the  worst-case necessary excess rate
is derived for any number of sources with any correlation matrix when space-time precoding derived from non-vanishing determinant codes is used. Further, we show that this bound can be significantly tightened for the case of uncorrelated sources.

\section{Problem Formulation and Background}
\label{sec:formulation}
In this section we provide the problem formulation
and briefly recall the achievable rates of IF source coding as developed in \cite{Integer-ForcingSourceCoding:OrdentlichErez2017}. We refer the reader to the latter  for an introduction and overview of IF source coding.

\subsection{Distributed Compression of Gaussian Sources}
We start by recalling the classical problem of distributed lossy compression of jointly Gaussian real random variables under a quadratic distortion measure. Specifically, we consider a distributed source coding setting with $K$ encoding terminals
and one decoder. Each of the $K$ encoders has access to a vector
of $n$ i.i.d. realizations of the random variable $x_k$, $k=1,..., K$.\footnote{The time axis will be suppressed in the sequel and vector notation will be reserved to
describe samples taken from different sources.}
The random vector ${\bf x} = (x_1, \ldots, x_K)^T$ (corresponding to the different sources)
is assumed to be Gaussian with zero mean and covariance matrix $\svv{K}_{\bf xx}\triangleq \mathbb{E}\left(\bf{x}\bf{x}^T\right)$.

Each encoder maps its observation $\bf{x}_k$ to an index using the encoding function
\begin{align}
\mathcal{E}_k~:~\mathbb{R}^n\rightarrow\left\{1,...,2^{nR_k}\right\},
\end{align}
and sends the index to the decoder.

The decoder is equipped with $K$ decoding functions
\begin{align}
\mathcal{D}_k~:~\left\{1,...,2^{nR_1}\right\}\times\cdots\times\left\{1,...,2^{nR_K}\right\}\rightarrow\mathbb{R}^n,
\end{align}
where $k=1,...,K$. Upon receiving $K$ indices, one from each terminal, it generates the estimates
\begin{align}
\hat{\bf{x}}_k=\mathcal{D}_k\left(\mathcal{E}_1({\bf x}_1),...,\mathcal{E}_K({\bf x}_K)\right),~~k=1,...,K.
\end{align}
A rate-distortion vector $(R_1,...,R_K; d_1,...,d_K)$ is achievable if there exist encoding functions, $\mathcal{E}_1,...,\mathcal{E}_K$, and decoding functions, $\mathcal{D}_1,...,\mathcal{D}_K$, such that $\frac{1}{n}\mathbb{E}\left(\|{\bf x}_k-\hat{\bf x}_k\|^2\right)\leq d_k$,
for all $k=1,...,K$.

We focus on the symmetric case where
$d_1=\cdots=d_K=d$
and
$R_1=\cdots=R_K=R/K$ where we denote the sum rate by $R$. The best known achievable scheme (for this symmetric setting; see \cite{Integer-ForcingSourceCoding:OrdentlichErez2017}, Section~I) is that of Berger and Tung  \cite{Multiterminalsourcecodig:Tung1978},
for which the following (in general, suboptimal) sum rate is achievable
\begin{align}
\sum_{k=1}^{K}R_k&\geq\frac{1}{2}\log\det{ \left(\svv{I}+\frac{1}{d}\svv{K}_{\bf xx}\right)} \nonumber \\
&\triangleq R_{\rm BT}.
\end{align}
As shown in \cite{Integer-ForcingSourceCoding:OrdentlichErez2017}, $R_{\rm BT}$ is a lower bound on the achievable rate of IF source coding. We will refer to $R_{\rm BT}$ as the Berger-Tung benchmark. To simplify notation, we note that $d$ can be ``absorbed'' into $\svv{K}_{\bf xx}$. Hence, without loss of generality, we assume throughout  that $d=1$.

\subsection{Compound Source Model And Scheme Outage
Formulation}
Consider distributed lossy compression of a vector of  Gaussian sources
\begin{align}
\svv{x}\sim\mathcal{N}(0,\svv{K}_{\bf xx}).
\end{align}

We define the following compound class of Gaussian sources, having the same value of $R_{\rm BT}$, via their covariance matrix:
\begin{align}
\mathcal{K}(R_{\rm BT}) = \left\{\svv{K}_{\bf xx}\in\mathbb{R}^{K\times K}: \frac{1}{2}\log\det\left(\svv{I}+\svv{K}_{\bf xx}\right)=R_{\rm BT}\right\}.
\end{align}

We quantify the  measure of the set of source covariance matrices by considering outage events, i.e., those events (sources) where integer forcing fails to achieve the desired level of distortion even though the rate
exceeds $R_{\rm BT}$.
More broadly, for a given quantization scheme, denote the necessary rate to achieve $d=1$ for a given covariance matrix $\svv{K_{xx}}$ as $R_{\rm scheme}(\svv{K_{xx}})$.
Then, given a target rate $R>R_{\rm BT}$ and a covariance matrix $\svv{K_{xx}}\in\mathcal{K}(R_{\rm BT})$, a scheme outage  occurs when $R_{\rm scheme}(\svv{K_{xx}})>R$.

To quantify the measure of ``bad" covariance matrices, we follow \cite{OutageBehaviorOfRandomlyPrecoded:Domanovitz2015} and apply a random orthonormal precoding matrix to the  (vector of) source samples prior to  encoding. As mentioned above, this amounts to joint processing of the samples
and
hence the problem is no longer distributed in general. Nonetheless, as in the scenario described in Section~\ref{sec:cran}, in certain statistical settings, this precoding operation is redundant as it can be viewed as being performed by nature.

Applying a precoding matrix to the source vector, we obtain a transformed source vector
\begin{align}
    \tilde{\bf x}=\svv{P}{\bf x},
\end{align}
with covariance matrix
\begin{align}
    \svv{K}_{\tilde x \tilde x}&=\svv{P}\svv{K}_{\bf x x}\svv{P}^T.
\end{align}

It follows that the achievable rate of a quantization scheme for the  precoded source is $R_{\rm scheme}(\svv{K}_{\tilde {\bf x} \tilde {\bf x}})$. When $\svv{P}$ is drawn at random, the latter rate is also random.
The worst-case (WC) scheme outage probability is defined in turn as
\begin{align}
    &  P^{\rm WC}_{\rm out,scheme}\left(R_{\rm BT},\Delta R\right)  \nonumber \\
      & = \sup_{\svv{K}_{\bf xx}\in\mathcal{K}(R_{\rm BT})}\prob\left(R_{\rm scheme}(\svv{K}_{\tilde {\bf x} \tilde {\bf x}})>
     R_{\rm BT}+\Delta R\right),
     \label{eq:P_WC_OUT}
\end{align}
where the probability is over the ensemble of precoding matrices considered and $\Delta R$ is the gap to the Berger-Tung benchmark.

In the sequel, we quantify the tradeoff between the quantization rate $R$ (or equivalently, between the excess rate $\Delta R=R-R_{\rm BT}$) and the outage probability $P^{\rm WC}_{\rm out,IF}\left(R_{\rm BT},\Delta R\right)$ as defined in (\ref{eq:P_WC_OUT}).

\subsection{Integer-Forcing Source Coding}
In a manner similar to IF equalization for channel coding, IF can be applied to the problem of distributed lossy compression. The  approach is based on standard  quantization followed by lattice-based binning. However,  in the IF framework, the decoder first uses the bin indices for recovering linear combinations with integer coefficients of the quantized signals, and only then
recovers the quantized signals themselves.

For our purposes, it suffices to state only
the achievable sum rate of IF source coding. 
We refer the reader to  \cite{Integer-ForcingSourceCoding:OrdentlichErez2017}
for the derivation and proofs.

\begin{theorem}[ \cite{Integer-ForcingSourceCoding:OrdentlichErez2017}, Theorem~1]
For any distortion $d > 0$ and any full-rank integer matrix $\svv{A} = [{\bf a}_1\cdots{\bf a}_K]^T\in \mathbb{Z}^{K\times K}$, there exists a (sequence of)
nested lattice pair(s)  $\Lambda\in\Lambda_f$ such that IF source coding can
achieve any sum rate satisfying
\begin{align}
    R>R_{\rm IF}(\svv{K}_{\bf x x},d;\svv{A})\triangleq \frac{K}{2}\log\left(\max_{k=1,\ldots,K}{\bf a}_k^T\left(\svv{I}+\frac{1}{d}\svv{K}_{\bf x x}\right){\bf a}_k\right)
\end{align}
\end{theorem}
We further note that this sum rate is achieved via symmetric rate allocation, i.e., by allocating $R_k=R_{\rm IF}(\svv{K}_{\bf x x},d;\svv{A})/K$ bits to each of the $K$ encoders. 
Since we assume $d=1$, it follows that IF source coding can achieve any (sum) rate satisfying
\begin{align}
R&>R_{\rm IF}^{\rm opt}(\svv{K}_{\bf x x})\nonumber \\
& \triangleq R_{\rm IF}(\svv{K}_{\bf x x};\svv{A}^{\rm opt}) \nonumber \\
& = \frac{K}{2}\log\left(\min_{\substack{\svv{A}\in\mathbb{Z}^{K\times K} \\ \det{\svv{A}}\neq0}}\max_{k=1,...,K}{\bf a}_k^T\left(\svv{I}+\svv{K}_{\bf xx}\right){\bf a}_k\right),
\label{eq:R_IF(A)}
\end{align}
where $\svv{a}_k^T$ is the $k$th row of the integer matrix $\svv{A}$.

The matrix $\svv{I}+\svv{K}_{\bf xx}$ is symmetric and positive definite, and therefore it admits a Cholesky decomposition
\begin{align}
    \svv{I}+\svv{K}_{\bf xx}=\svv{F}\svv{F}^T.
    \label{eq:def_F}
\end{align}
With this notation, we have
\begin{align}
R_{\rm IF}^{\rm opt}(\svv{K}_{\bf x x})=\frac{K}{2}\log\left(\min_{\substack{\svv{A}\in\mathbb{Z}^{K\times K} \\ \det{\svv{A}}\neq0}}\max_{k=1,...,K}\|\svv{F}^T\bf{a}_k\|^2\right).
\label{eq:sum_rate_char}
\end{align}

Denote by $\Lambda(\svv{F}^T)$ the $K$-dimensional lattice spanned by the matrix $\svv{F}^T$, i.e.,
\begin{align}
\Lambda(\svv{F}^T)\triangleq\left\{\svv{F}^T\bf{a}:~a\in\mathbb{Z}^K \right\}.
\end{align}
Then the problem of finding the optimal matrix $\svv{A}$ is equivalent to finding the shortest set of $K$ linearly independent vectors in $\Lambda(\svv{F}^T)$.
In other words, the rate per encoder achieved by IF source coding can be expressed using the $k$th successive minimum of the lattice $\Lambda(\svv{F}^T)$, where we recall that in general, for a full-rank $K \times K$ matrix $\svv{G}$:
\begin{definition}(successive minima)
Let $\Lambda(\svv{G})$ be a lattice spanned by the full-rank matrix $\svv{G}\in\mathbb{R}^{K\times K}$. For $k = 1, ... , K$, we define the $k$'th successive minimum as
\begin{align}
\lambda_k(\svv{G})\triangleq\inf\left\{r:{\rm dim}\left({\rm span}\left(\Lambda(\svv{G})\cap\mathcal{B}_{K}(r)\right)\right)\geq k\right\}
\end{align}
where $\mathcal{B}_{K}(r)=\left\{{\bf x}\in\mathbb{R}^K:\|{\bf x}\|\leq r\right\}$ is the closed ball of radius $r$ around ${\bf 0}$. In words, the $k$-th successive minimum of a lattice is the minimal radius of a ball centered around ${\bf 0}$ that contains $k$ linearly independent lattice points.
\end{definition}

Thus, \eqref{eq:sum_rate_char} can be restated as saying that IF source coding can achieve any rate greater than
\begin{align}
R_{\rm IF}^{\rm opt}(\svv{K}_{\bf x x}) 
&=\frac{K}{2}\log\left(\lambda_K^2(\svv{F}^T)\right).
\label{eq:RSumif}
\end{align}

Just as successive interference cancellation  significantly improves the achievable rate of  IF  equalization in channel coding (see, e.g., \cite{5757526}), an analogous scheme can be implemented in the case of IF source coding. Specifically, for a given full-rank integer matrix $\svv{A}$,  let $\svv{L}$ be defined
by the Cholesky
decomposition
\begin{align}
    \svv{A}\left(\svv{I}+\svv{K}_{\bf xx}\right)\svv{A}^T=\svv{L}\svv{L}^T
\end{align}
and denote the $m$th element of the diagonal of $\svv{L}$ by $\ell_{m,m}$.
Then, as shown in \cite{UniversalPrecodingforParallelGaussianChannels:Fischler2014} and \cite{he2016integer}, the achievable rate of successive IF source coding (which we denote as IF-SUC) for this choice of $\svv{A}$ is given by
\begin{align}
    R_{\rm IF-SUC}(\svv{K}_{\bf x x};\svv{A})=K \cdot \max_{k=1, \ldots, K} R_{{\rm IF-SUC},k}(\svv{K}_{\bf x x};\svv{A}),
\end{align}

where
\begin{align}
    R_{{\rm IF-SUC},k}(\svv{K}_{\bf x x};\svv{A})=\frac{1}{2}\log \left( \ell_{k,k}^2 \right).
    \label{eq:R_IFSUC_m}
\end{align}
We refer the reader to pages 36-37 in \cite{UniversalPrecodingforParallelGaussianChannels:Fischler2014} for details.
Finally, by optimizing over the choice of $\svv{A}$, we obtain
\begin{align}
R_{{\rm IF-SUC}}^{\rm opt}(\svv{K}_{\bf x x}) &\triangleq R_{{\rm IF-SUC}}(\svv{K}_{\bf x x};\svv{A}^{\rm opt}_{\rm SUC}) \nonumber \\
&=  \min_{\substack{\svv{A}\in\mathbb{Z}^{K\times K} \\ \det{\svv{A}}\neq0}}
R_{\rm IF-SUC}(\svv{K}_{\bf x x};\svv{A}).
\end{align}

While the gap between $R_{\rm IF}$ (and even more so $R_{\rm IF-SUC}$) and $R_{\rm BT}$ is quite small for most covariance matrices, it can nevertheless be arbitrarily large. We next quantify the measure of bad covariance matrices by considering randomly-precoded IF source coding.
\section{CRE-Precoded IF Source Coding and its Empirical Performance}
\label{sec:cre}

Recalling (\ref{eq:R_IF(A)}), and with a slight abuse of notation,  the rate of IF source coding for a given precoding matrix $\svv{P}$ is denoted by
\begin{align}
    &R_{\rm IF}^{\rm opt}(\svv{K}_{\bf x x},\svv{P}) \triangleq  R_{\rm IF}^{\rm opt}(\svv{P}\svv{K}_{\bf x x} \svv{P}^T) \nonumber \\
    &=\frac{K}{2}\log\left(\min_{\substack{\svv{A}\in\mathbb{Z}^{K\times K} \\ \det{\svv{A}}\neq0}}\max_{k=1,...,K}{\bf a}_k^T\left(\svv{I}+\svv{P}\svv{K}_{\bf x x}\svv{P}^T\right){\bf a}_k\right).
    \label{eq:R_if_p}
\end{align}

Since $\svv{K_{\bf x x}}$ is symmetric, it allows orthonormal diagonalization
\begin{align}
\svv{I}+\svv{K_{\bf x x}}=\svv{U}\svv{D}\svv{U}^T.
\label{eq:KDconnection}
\end{align}
When unitary precoding is applied, we have
\begin{align}
    \svv{I}+\svv{P}\svv{K_{x x}}\svv{P}^T=\svv{P}\svv{U}\svv{D}\svv{U}^T\svv{P}^T.
\end{align}

To quantify the measure of ``bad'' sources, we consider precoding matrices
that are uniformly (Haar) distributed over the group of orthonormal matrices. Such a matrix ensemble is referred to as CRE and is defined by the unique distribution on orthonormal matrices that is invariant under left and right orthonornal transformations \cite{metha1967random}. That is, given a random matrix $\svv{P}$ drawn from the CRE, for any orthonormal matrix $\grave{\svv{U}}$, both ${\svv{P}}\grave{\svv{U}}$ and $\grave{\svv{U}}{\svv{P}}$ are equal in distribution to ${\svv{P}}$.
Since $\svv{P}\svv{U}^T$ is equal in distribution to ${\svv{P}}$ for CRE precoding, for the sake of computing outage probabilities, we may simply assume that $\svv{U}^T$ (and also $\svv{U}$) is drawn from the CRE.

For a specific choice of integer vector ${\bf a_k}$, we define (again, with a slight abuse of notation)
\begin{align}
 R_{\rm IF}(\svv{D},\svv{U};{\bf a}_k)& \triangleq \frac{1}{2}\log\left({\bf a}_k^T\svv{U}\svv{D}\svv{U}^T {\bf a}_k\right) \nonumber \\
 &=\frac{1}{2}\log\left(\|\svv{D}^{1/2}\svv{U}^T{\bf a}_k\|^2\right),
 \label{eq:specifica}
\end{align}
and correspondingly
\begin{align}
    &R_{\rm IF}^{\rm opt}(\svv{D},\svv{U}) \triangleq \nonumber \\
    &\frac{K}{2}\log\left(\min_{\substack{\svv{A}\in\mathbb{Z}^{K\times K} \\ \det{\svv{A}}\neq0}}\max_{k=1,...,K}
    \|\svv{D}^{1/2}\svv{U}^T{\bf a}_k\|^2
    \right).
    \label{eq:PIF_rate1}
\end{align}
Let $\Lambda$ be the lattice spanned by $\svv{G}=\svv{D}^{1/2}\svv{U}^T$.
Then (\ref{eq:PIF_rate1}) may be rewritten as
\begin{align}
R_{\rm IF}^{\rm opt}(\svv{D},\svv{U})=\frac{K}{2}\log\left(\lambda_K^2(\Lambda)\right).
\label{eq:PIF_rate2}
\end{align}

Let us denote the set of all diagonal matrices having the same value of $R_{\rm BT}$, i.e.,
\begin{align}
\mathbb{D}(R_{\rm BT},K) = \{\svv{D}~:~\det\left(\svv{D}\right)=2^{2R_{\rm BT}}\}
.
\label{eq:compoundD}
\end{align}
We may thus rewrite the worst-case outage probability of IF source coding, defined in (\ref{eq:P_WC_OUT}), as
\begin{align}
&P^{\rm WC}_{\rm out,IF}\left(R_{\rm BT},\Delta R\right)\nonumber \\
&=\sup_{\svv{D}\in\mathbb{D}(R_{\rm BT},K)}\prob \left( R_{\rm IF}^{\rm opt}(\svv{D},\svv{U})>R_{\rm BT}+\Delta R \right),
\label{eq:rewirte6withD}
\end{align}
where the probability is with respect to the random selection of  $\svv{U}$ that is drawn from the CRE.

To illustrate the worst-case performance of CRE-precoded IF, we present its  empirical performance for the case of a two-dimensional compound Gaussian source
vector, where the  outage probability (\ref{eq:rewirte6withD}) is computed via Monte-Carlo simulation.

Figure~\ref{fig:x1} depicts the  results
for different values of $R_{\rm BT}$. Rather than plotting the worst-case outage probability, its complement is depicted, i.e., we plot the probability that the rate of IF falls below $R_{\rm BT}+\Delta R$. 
As can be seen from the figure, the WC outage probability (as a function of $\Delta R$) converges to a limiting curve as $R_{\rm BT}$ increases.
\begin{figure}[h]
\begin{center}
\includegraphics[width=0.75\columnwidth]{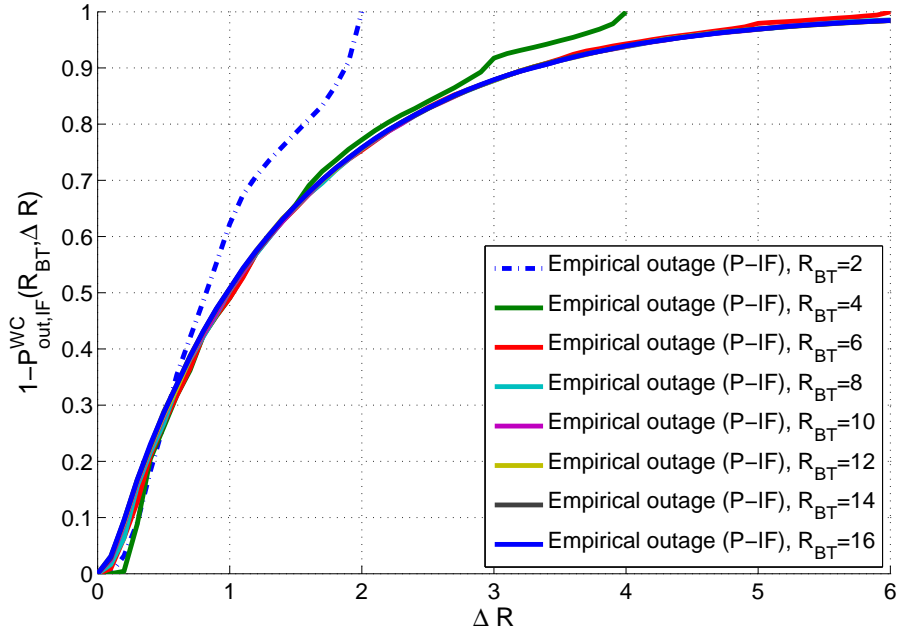}
\end{center}
\caption{Empirical results for (the complement of) the worst-case outage probability of IF source coding when applied to a two-dimensional compound Gaussian source vector as a function of $\Delta R$, for various values of $R_{\rm BT}$.}
\label{fig:x1}
\end{figure}

Figure~\ref{fig:x2} depicts the results for the single (high) rate  $R_{\rm BT}=16$. The required compression rate required to support a given worst-case outage probability constraint is marked, for several outage probabilities. 
We observe that:
\begin{itemize}
\item
For $10\%$ worst-case outage probability, a gap of  $\Delta R=3.292$ bits (or $1.646$ bits per source) is required.
\item
For $5\%$ worst-case outage probability, a gap of $\Delta R=4.293$ bits (or $2.1465$ bits per source) is required.
\item
For $1\%$ worst-case outage probability, a gap of $\Delta R=6.665$ bits (or $3.3325$ bits per source) is required.
\end{itemize}

\begin{figure}[h]
\begin{center}
\includegraphics[width=0.75\columnwidth]{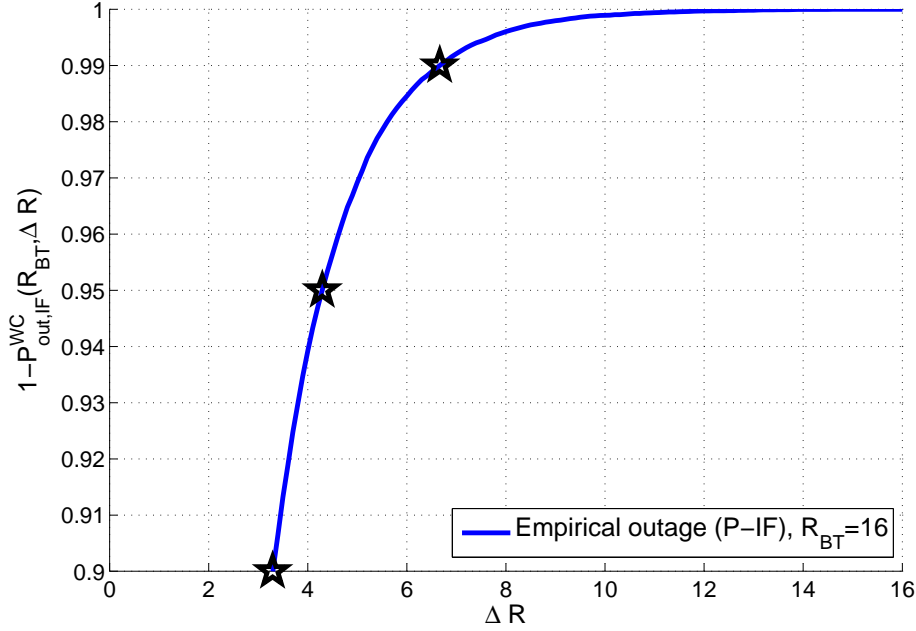}
\end{center}
\caption{Zoom in on the empirical worst-case outage probability for a two-dimensional compound Gaussian source vector with $R_{\rm BT}=16$.}
\label{fig:x2}
\end{figure}

\section{Upper Bounds on the Outage Probability for CRE-Precoded Integer-Forcing Source Coding}
\label{sec:upper}
In this section we develop achievability bounds for randomly-precoded IF source coding. As the derivation is very much along the lines of the results for the analogous problem in channel coding as developed in \cite{OutageBehaviorOfRandomlyPrecoded:Domanovitz2015}, we refer to results from the latter in many points.

The next lemma provides an upper bound on the outage probability of precoded IF source coding as a function of $R_{\rm BT}$ and the rate gap  $\Delta R$
(as well as the number of sources and $d_{\max}$ defined below). Denote
\begin{align}
    \mathbb{A}(\genGam,\delta,K)=\left\{{\bf a}\in \mathbb{Z}^{K}:0<\|{\bf a}\|<\sqrt{\frac{\genGam}{\delta}}\right\}.
    \label{eq:A_beta_d_general}
\end{align}
\begin{lemma}
\label{lem:lem1}
For any $K$ Gaussian sources such that $\svv{D}\in\mathbb{D}(R_{\rm BT},K)$, and for $\svv{U}$ drawn from the CRE, we have
\begin{align}
&\prob\left(R_{\rm IF}^{\rm opt}(\svv{D},\svv{U})<R_{\rm BT}+\Delta R\right) \nonumber \\
&\leq
\sum_{{\bf a}\in\mathbb{A}(\genGam,1/d_{\max},K)}\frac{K\alNt^{\frac{K-1}{2}}2^{-\frac{K-1}{K}(R_{\rm BT}+\Delta R)}}
{\frac{\aNorm^{K-1}}{2^{R_{\rm BT}}}{\sqrt{d_{\max}}}}
\end{align}
where $\displaystyle{d_{\max}=\max_{i}\svv{D}_{i,i}}$, the set ${\mathbb{A}(\genGam,1/d_{\max},K)}$ is defined in \eqref{eq:A_beta_d_general},
\label{eq:lemma2}
and where $\alNt$ is defined in \eqref{eq:alNtDef} below.
\end{lemma}
\begin{proof}
Let $\Lambda^*$ denote the dual lattice of $\Lambda$ and note that it is spanned by the matrix
\begin{align}
(\svv{G}^T)^{-1}=\svv{D}^{-1/2}\svv{U}^T.
\label{eq:dualLattice}
\end{align}
The successive minima of $\Lambda$ and $\Lambda^{*}$ are related by (Theorem 2.4 in \cite{KorkinZolotarev:Lagarias1990})
\begin{align}
\lambda_1(\Lambda^{*})^2\lambda_{K}(\Lambda)^2\leq\frac{K+3}{4}{\bar{\gamma}_{K}}^2,
\label{eq:Lagarias}
\end{align}
where
\begin{align}
\bar{\gamma}_{K} = \max \{ \gamma_i : 1 \leq i \leq K \}
\label{eq:bargamma}
\end{align}
with $\gamma_{K}$ denoting Hermite's constant.

The tightest known bound for Hermite's constant, as derived in \cite{minimumValueOfQuadraticForms:Blichfeldt1929}, is
 \begin{align}
\gamma_{K}\leq\left(\frac{2}{\pi}\right)\Gamma\left(2+\frac{K}{2}\right)^{2/K}.
\label{eq:blich}
\end{align}
Since this is an increasing function of $K$, it follows that $\bar{\gamma}_{K}$ is smaller than the r.h.s. of (\ref{eq:blich}).
Combining the latter with the exact values of
the Hermite constant for dimensions for which
it is known, we define
\begin{align}
\alNt =\begin{cases}
\frac{K+3}{4}\gamma_{K}^2,& K=1-8,24 \\
\frac{K+3}{4}\frac{2}{\pi}\Gamma\left(2+\frac{K}{2}\right)^{2/K},& {\rm otherwise}
\end{cases}.
\label{eq:alNtDef}
\end{align}

Therefore, we may bound the achievable rates of IF via the dual lattice as follows
\begin{align}
R_{\rm IF}^{\rm opt}(\svv{D},\svv{U})&\leq\frac{K}{2}\log\left({\alNt}\frac{1}{\lambda_1(\Lambda^{*})^2}\right).
\end{align}
Hence, we have
\begin{align}
    &\prob\left(R_{\rm IF}^{\rm opt}(\svv{D},\svv{U})>R_{\rm BT}+\Delta R\right) \nonumber \\
    & \leq \prob\left(\frac{K}{2}\log\left(\alNt\frac{1}{\lambda_1(\Lambda^{*})^2}\right)>R_{\rm BT}+\Delta R\right)  \nonumber \\
     & = \prob\left(\lambda_1(\Lambda^{*})^2<\alNt 2^{-\frac{2}{K}(R_{\rm BT}+\Delta R)}\right).
     \label{eq:PrLambda1Implicit-Comp}
\end{align}
Denote
\begin{align}
\genGam=\alNt 2^{-\frac{2}{K}(R_{\rm BT}+\Delta R)}.
\label{eq:beta}
\end{align}
We wish to bound~(\ref{eq:PrLambda1Implicit-Comp}), or equivalently, we wish to bound
\begin{align}
\prob\left(\lambda_1^2(\Lambda^{*})<{\genGam}\right)=\prob\left(\lambda_1(\Lambda^{*})<\sqrt{\genGam}\right)
\end{align}
for a given matrix $\svv{D} \in \mathbb{D}(R_{\rm BT},K)$.
Note that the event $\lambda_1(\Lambda^{*})<\sqrt{\genGam}$ is equivalent to the event
\begin{align}
\bigcup_{{\bf a}\in{\mathbb{Z}^{K}\setminus\{{\bf 0}\}}}||\svv{D}^{-1/2}\svv{U}^T\bf{a}||<\sqrt{\genGam}.
\end{align}
Applying the union bound yields
\begin{align}
\prob\left(\lambda_1(\Lambda^{*}) < \sqrt{\genGam}\right) \leq \sum_{{\bf a}\in\mathbb{Z}^{K}\setminus\{{\bf 0}\}}\prob\left(||\svv{D}^{-1/2}\svv{U}^T{\bf a}||<\sqrt{\genGam}\right).
\label{eq:outageProbReal2x2}
\end{align}
Note that whenever $\frac{||{\bf a}||}{\sqrt{d_{\max}}}
\geq \sqrt{\genGam}$, we have
\begin{align}
    \prob\left(||\svv{D}^{-1/2}\svv{U}^T{\bf a}||<\sqrt{\genGam}\right)=0.
    \label{eq:173}
\end{align}
Therefore, substituting $1/d_{\max}$ in (\ref{eq:A_beta_d_general}), the set of relevant vectors ${\bf a}$ is
\begin{align}
{\mathbb{A}(\genGam,1/d_{\max},K)}=\left\{{\bf a}\in \mathbb{Z}^K:0<||{\bf a}||<\sqrt{\genGam d_{\max}}\right\}.
\end{align}
It follows from (\ref{eq:outageProbReal2x2}) and (\ref{eq:173}) that
{\small
\begin{align}
\prob\left(\lambda_{1}(\Lambda^{*})<\sqrt{\genGam}\right)\leq \sum_{{\bf a}\in{\mathbb{A}(\genGam,1/d_{\max},K)}}\prob\left(\|\svv{D}^{-1/2}\svv{U}^T{\bf a}\|<\sqrt{\genGam}\right).
\label{eq:outageProbReal}
\end{align}
}%

The rest of the proof follows the footsteps of Lemma~2 in \cite{OutageBehaviorOfRandomlyPrecoded:Domanovitz2015} and is given in Appendix~{\ref{app:proofOfLemma1}}.
\end{proof}

While Lemma~\ref{lem:lem1} provides an explicit bound on the outage probability, in order to calculate it, one needs to go over all diagonal matrices in $\mathbb{D}(R_{\rm BT},K)$ and for each such diagonal matrix, sum over all the relevant integer vectors in $\mathbb{A}(\genGam,1/d_{\max},K)$.
Hence, the bound can be evaluated only for moderate compression rates and for a small number of sources. The following theorem, that may be viewed as the counterpart of Theorem~1 in \cite{OutageBehaviorOfRandomlyPrecoded:Domanovitz2015}, provides a looser, yet very simple
closed-form bound. Another advantage for this bound is that it does not depend on the Berger-Tung achievable rate.
\begin{theorem}
For any $K$ sources such that $\svv{D}\in\mathbb{D}(R_{\rm BT},K)$, and for $\svv{U}$ drawn from the CRE we have
\begin{align}
\prob\left(R_{\rm IF}^{\rm opt}(\svv{D},\svv{U})>R_{\rm BT}+\Delta R\right)
\leq
c(K)2^{-\Delta R},
\label{eq:thm1}
\end{align}
where
\begin{align}
    c(K)=K\alNt^{\frac{K}{2}}(K+c_{\max})\frac{\pi^{K/2}}{\Gamma(K/2+1)},
\end{align}
\begin{align}
\alNt =
\frac{K+3}{4}\frac{2}{\pi}\Gamma\left(2+\frac{K}{2}\right)^{2/K}
\end{align}
and
\begin{align}
    c_{\max}=
    \begin{cases}
    \left(2+\frac{\sqrt{K}}{2}\right)^{K}-\left(1-\frac{\sqrt{K}}{2} \right)^{K} & 1\leq K < 4 \\
    \left(1+\sqrt{K}\right)^{K}  & K \geq 4
    \end{cases}.
    \label{eq:ckdef}
\end{align}
Note that $c(K)$ is a constant that depends only on the number of sources $K$.
\label{thm:thm1}
\end{theorem}
\begin{proof}
See Appendix~\ref{app:proofOfThm1}.
\end{proof}


Similarly to the case of IF channel coding
(cf., Section~IV-C in \cite{OutageBehaviorOfRandomlyPrecoded:Domanovitz2015}),
analyzing Theorem~{\ref{thm:thm1}} reveals that there are two main sources for looseness that may be further tightened:
\begin{itemize}
\item
    \emph{Union bound} - While there is an inherent loss in the union bound, in fact, some terms in the summation (\ref{eq:outageProbReal}) may be completely dropped.\footnote{Similar to
    the derivation in Section~IV-C in
    \cite{OutageBehaviorOfRandomlyPrecoded:Domanovitz2015},
    a simple factor of $2$ can be deduced (regardless of the rate and number of sources) by noting that $\bf{a}$ and $-\bf{a}$ result in the same outcome and  hence there is no need to account for both cases.\label{foot:3}} Specifically, using  Corollary~1 in
    \cite{OutageBehaviorOfRandomlyPrecoded:Domanovitz2015},
    the set $\mathbb{A}(\genGam,1/d_{\max},K)$ appearing in the summation in (\ref{eq:lemma2}) may be replaced by the smaller set $\mathbb{B}(\genGam,1/d_{\max},K)$ where
\end{itemize}
\vspace*{-\baselineskip}
{\small
\begin{align}
&\mathbb{B}(\genGam,d,K)= \nonumber \\
&\left\{{\bf a}\in \mathbb{Z}^K:0<\|{\bf a}\|<\sqrt{\frac{\genGam}{d}} \:  {\rm and} \: \nexists 0<c<1 \: {\rm s.t.} \: c{\bf a}\in\mathbb{Z}^K \right\}.
\end{align}
}%
\vspace*{-\baselineskip}
\begin{itemize}
\item
    \emph{Dual Lattice} - Bounding via the dual lattice induces a loss reflected in (\ref{eq:Lagarias}). This may be circumvented for the case of a two-dimensional  source vector by using IF-SUC, as accomplished in Lemma~\ref{lem:lem2} and Theorem~\ref{thm:thm2} which we present next.
\end{itemize}

\begin{lemma}
\label{lemma2}
For a two-dimensional Gaussian  source vector such that $\svv{D}\in\mathbb{D}(R_{\rm BT},K)$, and for $\svv{U}$ drawn from the CRE, we have
\begin{align}
&\prob\left(R_{\rm IF-SUC}^{\rm opt}(\svv{D},\svv{U})>R_{\rm BT}+\Delta R\right) \nonumber \\
&\leq
\sum_{{\bf a}\in\mathbb{B}(\genGam,d_{\min},K)}\frac{2\sqrt{\genGam}}{\|{\bf a}\|2^{R_{\rm BT}}\frac{1}{\sqrt{d_{\min}}}}
\end{align}
where $\displaystyle{d_{\min}=\min_{i}\svv{D}_{i,i}}$ and
${\mathbb{A}(\genGam,d_{\min},K)}$ is defined in \eqref{eq:A_beta_d_general},
\label{lem:lem2}
\end{lemma}

\begin{proof}
See Appendix~\ref{app:proofOfThm2}.
\end{proof}

\begin{theorem}
For a two-dimensional Gaussian  source vector such that $\svv{D}\in\mathbb{D}(R_{\rm BT},K)$, and for $\svv{U}$ drawn from the CRE, we have
\begin{align}
\prob\left(R_{\rm IF-SUC}^{\rm opt}(\svv{D},\svv{U})>R_{\rm BT}+\Delta R\right) \leq
c'(K)2^{-\Delta R},
\end{align}
where
\begin{align}
    c'(K)=
    2\pi\left(5+3\sqrt{2}
    \right).
\end{align}
\label{thm:thm2}
\end{theorem}
\begin{proof}
See Appendix~\ref{app:proofOfThm2}.
\end{proof}

\begin{figure}
\begin{center}
\includegraphics[width=0.75\columnwidth]{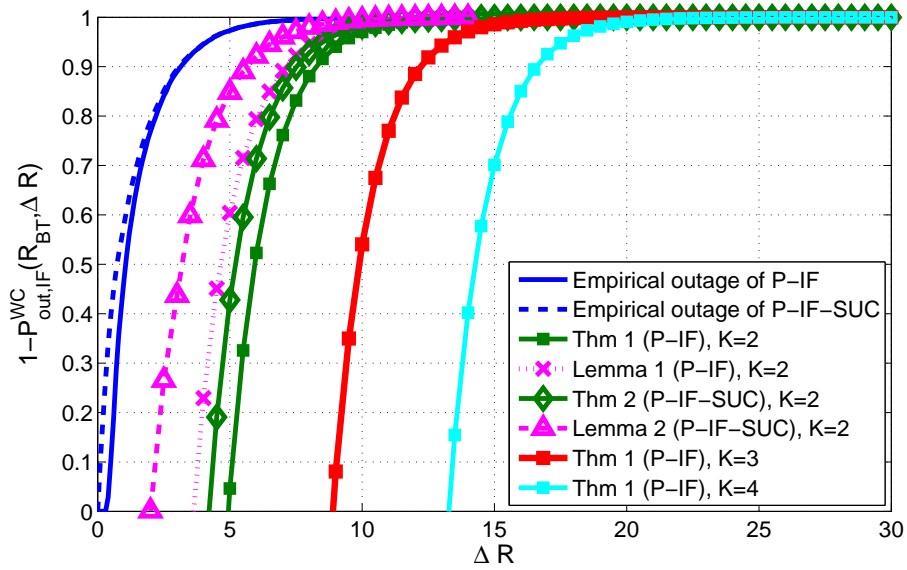}
\end{center}
\caption{Upper bounds on the outage probability of IF source coding for various values of source dimension $K$.}
\label{fig:xx2}
\end{figure}

Figure~\ref{fig:xx2} depicts the bounds derived as well as results of a Monte Carlo evaluation of (\ref{eq:P_WC_OUT}) for the case of a two-dimensional Gaussian and CRE precoding.\footnote{The bounds are computed after applying a factor of $2$ to the lemmas in accordance to footnote~\ref{foot:3}. Again, rather than plotting the worst-case outage probability, we plot its complement.} When calculating the empirical curves and the lemmas, we assumed  high quantization rates ($R_{\rm BT}=14$). The lemmas were calculated by going over a grid of values of $d_1$ and $d_2$ satisfying
$d_1 d_2=2^{2 \cdot 14}$.

\subsection{Application: Distributed Compression for Cloud Radio Access Networks}
\label{sec:cran}

Since we described IF source coding as well as the precoding over the reals, we outline the application of IF source coding for the cloud  radio access network (C-RAN) scenario assuming a real channel model. We then comment on the adaptation of the scheme to the more realistic scenario of a complex channel.

Consider the C-RAN scenario depicted in
Figure~\ref{fig:C-ran topology}
where $M$ transmitters send their data (that is modeled as an i.i.d. Gaussian source vector) over a $K\times M$ MIMO broadcast channel $\svv{H}\in\mathbb{R}^{K\times M}$.
The data is received at $K$ receivers (relays) that wish to compress and forward it for processing (decoding) at a central node via rate-constrained noiseless bit pipes.

\begin{figure}
\begin{center}
\includegraphics[width=0.75\columnwidth]{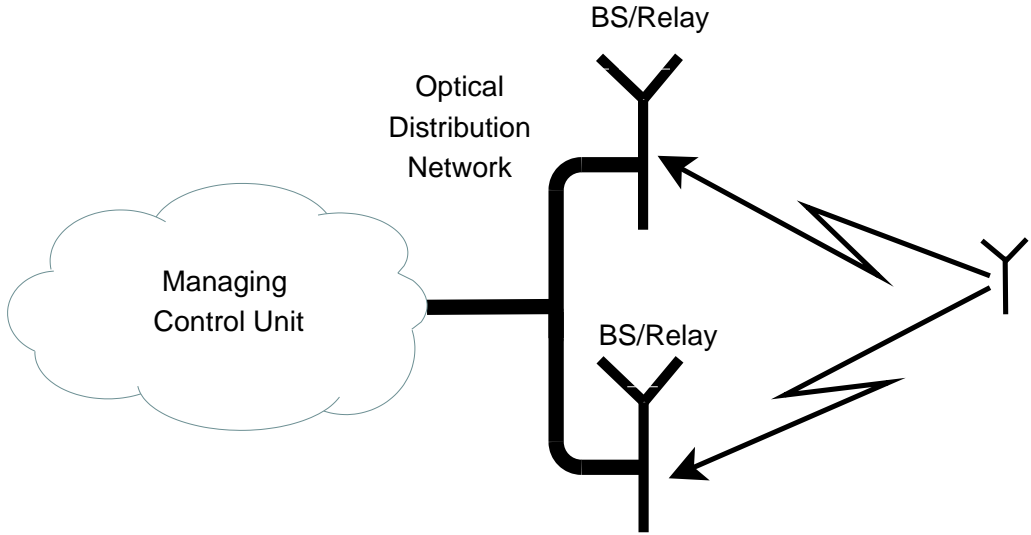}
\end{center}
\caption{Cloud radio access  network communication scenario. Two relays compress and forward the correlated signals they receive from several users.}
\label{fig:C-ran topology}
\end{figure}

As we wish to minimize the distortion at the central node subject to the rate constraints, this is a distributed lossy source coding problem. See depiction in Figure~\ref{fig:C-ran topology}.

Here, the covariance matrix of the received signals at the relays is given by
\begin{align}
\svv{K}_{\bf xx}=\SNR\svv{H}\svv{H}^T+\svv{I}.
\end{align}
We note that we can ``absorb'' the $\SNR$ into the channel and hence we set $\SNR=1$, so that
\begin{align}
\svv{K}_{\bf xx}=\svv{H}\svv{H}^T+\svv{I}.
\end{align}
We further assume that the entries of the channel matrix  $\svv{H}$ are Gaussian i.i.d., i.e. $\forall_{i,j}\svv{H}_{i,j}\sim \mathcal{N}(0,\sigma^{2})$. As mentioned in the introduction, the SVD of this matrix
\begin{align}
\svv{H}=\svv{\tilde{U}}\svv{\tilde{\Sigma}}\svv{\tilde{V}}^T
\end{align}
satisfies that  $\svv{\tilde{U}}$ and $\svv{\tilde{V}}$ belong to the CRE. We may therefore express the (random) covariance matrix as
\begin{align}
\svv{K}_{\bf xx}=\svv{\tilde{U}}(\svv{\tilde{\Sigma}}+\svv{I})(\svv{\tilde{\Sigma}}+\svv{I})^T\svv{\tilde{U}}^T,
\label{eq:svdOfKxx}
\end{align}
where $\svv{\tilde{U}}$ is drawn from the CRE.
It follows that the precoding matrix $\svv{P}$ is redundant (as we assumed that $\svv{P}$ is also drawn from the CRE). Thus, the analysis above holds also for the considered scenario. 

Specifically, assuming the encoders are subject to an equal rate constraint, 
then for a given distortion level,  the relation between the compression rate of IF source coding and the guaranteed outage probability (for meeting the prescribed distortion) is bounded using Theorem~\ref{thm:thm1} above.

We note that just as precoded
IF channel coding can be applied to complex channels as
described in  \cite{OutageBehaviorOfRandomlyPrecoded:Domanovitz2015},
so can precoded IF source coding be extended to complex Gaussian sources. In describing an outage event in this case we assume that the precoding matrix is drawn from the circular unitary ensemble (CUE). The bounds derived above (replacing $K$ with $2K$ in all derivations) for the relation between the compression rate of IF source coding and the worst-case outage probability  hold for the C-RAN scenario over complex Gaussian channels $\svv{H}\in\mathbb{C}^{K\times M}$, where the CUE precoding can be viewed as been performed by nature.

\section{Performance Guarantees for Integer-Forcing Source Coding with Deterministic Precoding}
\label{sec:additive}
In this section, we consider the performance of IF source coding when used
in conjunction with judiciously chosen (deterministic) precoding.
Worst-case performance will be measured in a stricter sense than in previous sections; namely, no outage is allowed. We begin by establishing an additive bound applicable for general Gaussian sources, in the form of a constant gap from the Berger-Tang benchmark. Similar to the bound established in  \cite{ordentlich2015precoded} for IF channel coding (specifically, Theorem 5 in the latter), the derived gap depends only on the number of sources and the properties of the non-vanishing code which is used as the underlying universal transformation. We note that the derived bound on the gap is even larger than that in the channel coding  counterpart and thus its applicability is limited. The reason for the difference in the derived gap is that unlike in the bound of \cite{ordentlich2015precoded}, we were 
unable to use the transference theorem of Banaszczyk \cite{Banaszczyk} and thus resorted to Minkowski's theorem (as recalled in Appendix C, Theorem~\ref{thm:Minkowski}). 

Then, in Section~\ref{sec:uncorr}, we consider the special case  of independent Gaussian sources having different variances. Here, we use a very different approach for analysis, by which we are able to  establish a much tighter bound on the gap to the Berger-Tung benchmark. The derived performance guarantees are tight enough  to be useful in practical scenarios, at least for a small number sources. 

\subsection{Additive Bound for General Sources}
Similar to the case of channel coding, we can derive a worst-case additive bound for the gap to the Berger-Tung benchmark.
Achieving this guaranteed performance requires joint algebraic number-theoretic based space-time precoding at the encoders. 
The following theorem is due to
Or Ordentlich \cite{OrdentlichPrivate2017}.
\begin{theorem}[Ordentlich]
\label{thm:Or}
For any $K$ sources with covariance matrix $\svv{K}_{\bf xx}$ and Berger-Tung benchmark $R_{\rm BT}$, the excess rate with respect to the Berger-Tung benchmark  (normalized per the number of time-extensions used) of space-time IF source coding with an NVD precoding matrix with minimum determinant $\delta_{\min}$ is bounded by
\begin{align}
R_{\rm IF} -R_{\rm BT} \leq 2K^3\log(2K^2)+K^2\log\frac{1}{\delta_{\min}}.
\end{align}
\label{thm:thm3}
\end{theorem}
\begin{proof}
See Appendix~\ref{app:additiveBound}.
\end{proof}

We note that the gap to the Berger-Tung benchmark is large (even larger than the one derived for channel coding IF in \cite{ordentlich2015precoded}) and is thus of limited applicability. Nevertheless, we first note that although the guaranteed gap is large, numerical evidence indicates that true gap is most likely much smaller. Hence, we believe the bound may be significantly tightened. 
Furthermore, if we relax the no outage restriction, the much tighter bounds of  Section~\ref{sec:upper} (using random precoding drawn from the CRE) are directly applicable.

\subsection{Uncorrelated Sources}
\label{sec:uncorr}
For the special case of uncorrelated Gaussian sources, much tighter bounds (in comparison to Theorem~\ref{thm:Or}) on the worst-case quantization rate of IF for a given distortion level may be obtained.
First, space-time precoding may be replaced with precoding over space only. This
allows to obtain a tighter counterpart to
Theorem~\ref{thm:Or}, as derived in Section~2.4.2 of \cite{UniversalPrecodingforParallelGaussianChannels:Fischler2014}.

We next derive yet tighter performance guarantees, also following ideas developed in \cite{UniversalPrecodingforParallelGaussianChannels:Fischler2014}, by numerically evaluating the performance of IF source coding over a ``densely'' quantized set of source (diagonal) covariance matrices belonging to the compound class, and then bounding the excess rate w.r.t. to the evaluated ones for any possible source vector in the compound class.

In the case of uncorrelated sources, the covariance matrix $\svv{K}_{\bf x x}$ is diagonal.
Hence, (\ref{eq:svdOfKxx}) becomes
\begin{align}
    \svv{K}_{\bf x x}&=\svv{\tilde{\Sigma}}\svv{\tilde{\Sigma}}^T \nonumber \\
    &\triangleq\svv{S},
\end{align}
where
\begin{align}
    \svv{S}=\begin{bmatrix}
    s_1^2 & 0 & \cdots & 0 \\
    0 & s_2^2 & \cdots & 0 \\
    \vdots & 0 & \ddots & \vdots \\
    0 & 0 & \cdots & s_K^2
    \end{bmatrix}.
\end{align}
We denote ${\bf s}=\diag(\svv{S})$. The compound set of sources may be parameterized by
\begin{align}
\mathcal{S}(R_{\rm BT}) = \left\{\svv{S}:\sum_{i=1}^{K} \frac{1}{2} \log \left(1 + s_{i}^2\right) =R_{\rm BT}\right\}.
\label{eq:D_R}
\end{align}


We note that we may
associate with each diagonal element a ``rate'' corresponding to an individual source
\begin{align}
R_i=\frac{1}{2} \log \left(1+s_i^2\right).
\label{eq:R_S}
\end{align}
Thus, the compound class of sources
may equivalently be represented by the
set of rates
\begin{align}
\mathcal{R}(R_{\rm BT})=\left\{\left(R_1,R_2,\ldots,R_K\right) \in \mathbb{R}^K:\sum_{i=1}^K R_i=R_{\rm BT} \right\}.
\end{align}

We define a  ``quantized"  rate-tuple set as follows. The interval $[0, R_{\rm BT}]$ is divided into $N$ sub-intervals, each of length $\Delta = R_{\rm BT}/N$. Thus, the resolution is determined by the parameter $\Delta$. The quantized rate-tuples belong to the grid
\begin{align}
\mathcal{R}^{\Delta}(R_{\rm BT})=\Bigl\{\left(R_1^{\Delta},R_2^{\Delta},\ldots,R_K^{\Delta}\right)\in\Delta \cdot R_{\rm BT}\cdot{\mathbb{Z}^{+}}^K:  \nonumber \\
\sum_{i=1}^K R_i^\Delta=R_{\rm BT} \Bigr \}.
\end{align}
We may similarly define the (non-uniformly)
quantized set $\mathcal{S}^{\Delta}(R_{\rm BT})$ of diagonal matrices such that the diagonal entries satisfy
$s_{i,\Delta}^2=2^{2R_i^{\Delta}}-1$, $i=1,\ldots,K$, where $\svv{R}^{\Delta}\in \mathcal{R}^{\Delta}(R_{\rm BT})$.



\begin{theorem}
For any Gaussian vector of independent sources with covariance matrix $\svv{S}$ such that  $\svv{S}\in \mathcal{S}(R_{\rm BT})$, the  rate of IF source coding with a given precoding matrix $\svv{P}$ is upper bound by
\begin{align}
    R_{\rm IF}^{\rm opt}(\svv{S},\svv{P})\leq
    \max_{\svv{S}^{\Delta} \in \mathcal{S}^\Delta(R_{\rm BT})  } R_{\rm IF}^{\rm opt}\left(\svv{S}^{\Delta},\svv{P}\right)+K\log\eta
\end{align}
where
\begin{align}
    \eta=\sqrt{\frac{2^{2\left(\frac{R_{\rm BT}}{K}\right)}-1}{2^{2\left(\frac{R_{\rm BT}}{K}-(K-1)\Delta R_{\rm BT}\right)}-1}}.
\label{eq:eta}
\end{align}
\label{thm4}
\end{theorem}
\begin{proof}

Assume we have a covariance matrix $\svv{S}$ in the compound class. Hence, its associated rate-tuple satisfies $(R_1, R_2, \ldots, R_K) \in \mathcal{R}(R_{\rm BT})$. Assume without loss of generality that
\begin{align}
    R_1 \leq R_2 \leq \ldots \leq R_K.
\end{align}




By \eqref{eq:specifica}, the rate of IF source coding associated with a specific  integer linear combination vector $\svv{a}$ is
\begin{align}
    R_{{\rm IF},k}(\svv{S},\svv{P};\svv{A})=\frac{1}{2}\log\left( {\bf a}_k^T\svv{P}\left(\svv{I}+\svv{S}\right)\svv{P}^T{\bf a}_k\right).
    \label{eq:quadratic}
\end{align}

Denote
\begin{align}
    {\bf v}={\bf a}_k^T\svv{P}.
    \label{eq:defOfv}
\end{align}
We will need the following two lemmas, whose proofs appear in Appendices \ref{app:proofOfLemma3} and
\ref{app:proofOfLemma4}, respectively.
\begin{lemma}
For any diagonal covariance matrix $\svv{S}$, associated with a rate-tuple  $(R_1, R_2, \ldots, R_K) \in \mathcal{R}(R_{\rm BT})$, there exists a diagonal covariance matrix  $\svv{S}^{\Delta}$, associated with a rate-tuple  $(R^{\Delta}_1, R^{\Delta}_2, \ldots, R^{\Delta}_K) \in \mathcal{R}^\Delta(R_{\rm BT})$, such that
\begin{align}
s_i \leq \eta^2 s_{i,\Delta}^2
\end{align}
for $1\leq i \leq K$, where $\eta$ is defined in (\ref{eq:eta}).
\label{lem:lem3}
\end{lemma}

\begin{lemma}
Consider a  Gaussian vector with a diagonal covariance matrix $\svv{S}$ and let $\svv{A}$ be an invertible integer matrix.
Then for any $\beta \geq 1$, we have
\begin{align}
    R_{{\rm IF},k}(\beta^2\svv{S},\svv{P};\svv{A})\leq \log\left(\beta\right)+ R_{{\rm IF},k}(\svv{S},\svv{P};\svv{A}).
\end{align}
\label{lem:lem4}
\end{lemma}


Using Lemma \ref{lem:lem3}, we denote by $\svv{S}^{\Delta}$ the covariance matrix associated with $\svv{S}$, and whose existence is guaranteed by the lemma. Recalling (\ref{eq:R_if_p}), we have
\begin{align}
    R_{\rm IF}^{\rm opt}({\svv{S}^{\Delta}},\svv{P})=\frac{K}{2}\log\left(\min_{\substack{\svv{A}\in\mathbb{Z}^{K\times K} \\ \det{\svv{A}}\neq0}}\max_{k=1,...,K} R_{\rm IF}({\svv{S}^{\Delta}},\svv{P};{\bf a}_k) \right).
\end{align}

Denoting $\svv{A}^{\rm opt}_{\Delta}$ as the optimal integer matrix for the quantized source $\svv{S}^{\Delta}$, it follows that
\begin{align}
R^{\rm opt}_{\rm IF}({\svv{S}},\svv{P}) &=
R_{\rm IF}({\svv{S}},\svv{P};\svv{A}^{\rm opt}) \nonumber \\
&\leq R_{\rm IF}({\svv{S}},\svv{P};\svv{A}^{\rm opt}_{\Delta}),
\label{eq:optimalityOfA}
\end{align}
where the inequality follows since $\svv{A}^{\rm opt}_{\Delta}$ is the optimal integer matrix for $\svv{S}^{\Delta}$ and not necessarily for $\svv{S}$. From \eqref{eq:quadratic} and by the definition of $\svv{v}$ in \eqref{eq:defOfv}, we have
\begin{align}
     R_{{\rm IF},k}( \eta^2{\svv{S}^{\Delta}},\svv{P};\svv{A}^{\rm opt}_{\Delta}) & =\frac{1}{2} \log\left(\sum_{i=1}^K { v}_i^2(1+\eta^2{s_{i,\Delta}^2})\right) \nonumber \\
     & \geq \frac{1}{2}\log\left(\sum_{i=1}^K { v}_i^2(1+s_{i}^2)\right) \nonumber \\
     & =  R_{{\rm IF},k}(\svv{S},\svv{P};\svv{A}^{\rm opt}_{\Delta}).
     \label{eq:75}
\end{align}
Using Lemma~\ref{lem:lem4}, we further have that
\begin{align}
R_{{\rm IF},k}( \eta^2{\svv{S}^{\Delta}},\svv{P};\svv{A}^{\rm opt}_{\Delta}) \leq \frac{1}{2}\log(\eta^2)+R_{{\rm IF},k}({\svv{S}^{\Delta}},\svv{P};\svv{A}^{\rm opt}_{\Delta}).
\label{eq:76}
\end{align}
Combining (\ref{eq:75}) and (\ref{eq:76}), we obtain
\begin{align}
    R_{{\rm IF},k}(\svv{S},\svv{P};\svv{A}^{\rm opt}_{\Delta}) & \leq R_{\rm IF}( \eta^2{\svv{S}^{\Delta}},\svv{P};\svv{A}^{\rm opt}_{\Delta}) \nonumber \\
    & \leq \log(\eta)+R_{{\rm IF},k}({\svv{S}^{\Delta}},\svv{P};\svv{A}^{\rm opt}_{\Delta}).
    \label{eq:78}
\end{align}
We therefore have
(\ref{eq:78}) we have
\begin{align}
R_{\rm IF}({\svv{S}},\svv{P};\svv{A}^{\rm opt}_{\Delta}) & = K \max_{k=1,...,K} R_{{\rm IF},k}({\svv{S}},\svv{P};\svv{A}^{\rm opt}_{\Delta}) \nonumber \\
& \leq K \max_{k=1,...,K} \bigl( \log(\eta)+R_{{\rm IF},k}({\svv{S}^{\Delta}},\svv{P};\svv{A}^{\rm opt}_{\Delta})\bigr) \nonumber \\
& = K\log(\eta) + R_{\rm IF}(\svv{S}^{\Delta},\svv{P};\svv{A}^{\rm opt}_{\Delta}).
\end{align}
Thus, recalling \eqref{eq:optimalityOfA}, we conclude that
\begin{align}
    R_{\rm IF}({\svv{S}},\svv{P};\svv{A}^{\rm opt}) \leq R_{\rm IF}(\svv{S}^{\Delta},\svv{P};\svv{A}^{\rm opt}_{\Delta}) + K\log(\eta)
\end{align}
and therefore, for any $\svv{S}\in \mathcal{S}(R_{\rm BT})$, we have
\begin{align}
R_{\rm IF}^{\rm opt}({\svv{S}},\svv{P}) \leq \max_{\svv{S}^{\Delta}\in\mathbb{S}^{\Delta}}R_{\rm IF}^{\rm opt}(\svv{S}^{\Delta},\svv{P}) + K\log(\eta).
\end{align}
This concludes the proof of
the theorem.

\end{proof}

As an example for the achievable performance, Figure~\ref{fig:diagSource_boundFromAnt} gives the empirical worst-case performance for two and three (real) sources that is achieved when using IF source coding and a fixed precoding matrix over the grid $\mathcal{R}^{\Delta}(R_{\rm BT})$ for $\Delta=0.01$. The precoding matrix was taken from \cite{oggier2004algebraic}. The explicit precoding matrix used for two sources is $\svv{P}=\svv{\rm cyclo_2}$ where
\begin{align}
    \svv{\rm cyclo_2}=\begin{bmatrix}
    -0.5257311121 & -0.8506508083 \\
    -0.8506508083 & 0.5257311121
    \end{bmatrix},
    \label{eq:cyclo}
\end{align}
and for three sources $\svv{P}=\svv{\rm cyclo_3}$ where
\begin{align}
    \svv{\rm cyclo_3}=\begin{bmatrix}
    -0.3279852776 & -0.5910090485 & -0.7369762291 \\
    -0.7369762291 & -0.3279852776 & 0.5910090485 \\
    -0.5910090485 & 0.7369762291 & -0.3279852776
    \end{bmatrix}.
    \label{eq:cyclo3}
\end{align}

Rather than plotting the gap from the Berger-Tung benchmark, we plot the efficiency $\frac{R_{IF}}{R_{\rm BT}}$, i.e., the ratio of the (worst-case) rate of IF source coding and $R_{\rm BT}$. We also plot the upper bound on $R_{\rm IF}(\svv{S},\svv{P})$ given by Theorem~\ref{thm4}. As is apparent from Figure~\ref{fig:diagSource_boundFromAnt}, the guaranteed efficiency approaches $1$ quite rapidly as the quantization rate grows.


\begin{figure}[h]
\begin{center}
\includegraphics[width=0.75\columnwidth]{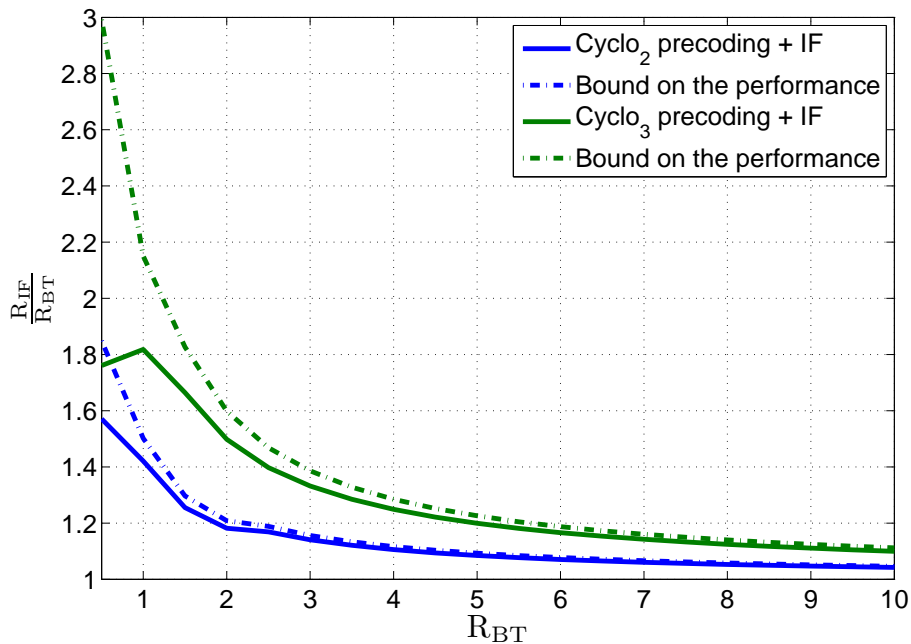}
\end{center}
\caption{Empirical and guaranteed (upper bound) worst-case efficiency of IF source coding for two and three uncorrelated Gaussian sources, when using the precoding matrices $\svv{\rm cyclo_2}$ and $\svv{\rm cyclo_3}$ given in \eqref{eq:cyclo}-\eqref{eq:cyclo3}, respectively,  and taking $\Delta=0.01$ for the calculation of Theorem~\ref{thm4}.}
\label{fig:diagSource_boundFromAnt}
\end{figure}

\section{Summary}
It has been observed in previous works that integer-forcing source coding is a very effective method for compression of distributed Gaussian sources, for ``most'' but not all source covariance matrices. 
In the present work we have  quantified the measure of bad covariance matrices by characterizing the probability that integer-forcing source coding fails as a function of the allocated rate, where the probability is with respect to a random orthonormal transformation that is applied to the sources prior to quantization. This characterization is directly applicable to the case where the signals to be compressed correspond to the antenna inputs of relays in an i.i.d. Rayleigh fading environment, as in such a scenario the transformation can be viewed as being done by nature. 

Integer-forcing source coding is also useful in a non-distributed scenario, where the integer-forcing scheme offers a universal, yet practical, compression method using equal-rate quantizers. Specifically, quantization is performed after passing the sources through an algebraic number theoretic based space-time transformation, followed by a modulo operation.
We have obtained constant-gap universal bounds on the maximal possible rate loss when integer-forcing source is used in such a scenario.
Further, we have shown that when it is known that the sources are uncorrelated but may be of different variances, applying a space-only transformation is sufficient to attain universal performance guarantees and moreover that the resulting bound on the maximal possible rate loss is much smaller than for the general case. An interesting avenue for further research is to try to generalize the numerical bounding technique applied in Section ~\ref{sec:uncorr} to obtain tighter bounds for the performance of precoded integer-forcing source coding of general Gaussian sources.

\appendices
\section{Proof of Theorem~\ref{lem:lem1}}
\label{app:proofOfLemma1}
Following the footsteps of Lemma~2 in \cite{OutageBehaviorOfRandomlyPrecoded:Domanovitz2015}  and adopting the same geometric interpretation described there, we may interpret
$\prob\left(\|\svv{D}^{-1/2}\svv{U}^T{\bf a}\|<\sqrt{\genGam}\right)$ as the ratio of the surface area of an ellipsoid that is inside a ball with radius $\sqrt{\genGam}$ and the surface area of the entire ellipsoid.
The axes of this ellipsoid are defined by $x_i=\frac{\|{\bf a}\|}{\sqrt{d_i}}$.

Denote the vector ${\bf o}_{\|{\bf a}\|}$ as a vector drawn from the CRE with norm $\|{\bf a}\|$. Using Lemma~1 in \cite{OutageBehaviorOfRandomlyPrecoded:Domanovitz2015} and since we assume that $\svv{U}^T$ is drawn from the CRE, we have that the right hand side of (\ref{eq:outageProbReal}) is equal to
\begin{align}
&\sum_{{\bf a}\in{\mathbb{A}(\genGam,1/d_{\max},K)}}\prob\left(\|\svv{D}^{-1/2}{\bf o}_{\|{\bf a}\|}<\sqrt{\genGam}\right) \nonumber \\
&=\sum_{{\bf a}\in{\mathbb{A}(\genGam,1/d_{\max},K)}}\frac{\rm CAP_{\rm ell}(x_1,\cdots,x_K)}{L(x_1,\cdots,x_K)}
\label{eq:ratioForSUC}
\end{align}
where
\begin{align}
{\rm CAP_{\rm ell}(x_1,\cdots,x_K)}<K\frac{\pi^{K/2}}{\Gamma(1+K/2)}\sqrt{\genGam}^{K-1}\triangleq\overline{\rm CAP_{\rm ell}},
\label{eq:CAPellSUC}
\end{align}
and
\begin{align}
{L(x_1,\cdots,x_K)}&>\frac{\pi^{K/2}}{\Gamma(1+K/2)}\frac{\aNorm^{K}}{\prod_{i=1}^{K}\sqrt{d_i}}\sum_{i=1}^{K}\frac{\sqrt{d_i}}{\aNorm}\\
&>\frac{\pi^{K/2}}{\Gamma(1+K/2)}\frac{\aNorm^{K-1}{\sqrt{d_{\max}}}}{2^{R_{\rm BT}}}\triangleq\underline{L}.
\label{eq:ellCirSUC}
\end{align}
Substituting~(\ref{eq:CAPellSUC}) and (\ref{eq:ellCirSUC}) in (\ref{eq:ratioForSUC}), we obtain
\begin{align}
& \sum_{{\bf a}\in{\mathbb{A}(\genGam,1/d_{\max},K)}}\frac{\rm CAP_{\rm ell}(x_1,\cdots,x_K)}{L(x_1,\cdots,x_K)} \nonumber \\
& <  \sum_{{\bf a}\in{\mathbb{A}(\genGam,1/d_{\max},K)}}\frac{\overline{\rm CAP_{\rm ell}}}{\underline{L}} \nonumber \\
&= \sum_{{\bf a}\in{\mathbb{A}(\genGam,1/d_{\max},K)}}\frac{K{\sqrt{\genGam}}^{K-1}}{\frac{\aNorm^{K-1}}{2^{R_{\rm BT}}}{\sqrt{d_{\max}}}}.
\end{align}
Recalling (see \eqref{eq:beta}) that $\genGam=\alNt 2^{-\frac{2}{K}(R_{\rm BT}+\Delta R)}$, we obtain
\begin{align}
&\prob\left(R_{\rm IF-SUC}^{\rm opt}(\svv{D},\svv{U})>R_{\rm BT}+\Delta R\right) \nonumber \\
&< \sum_{{\bf a}\in\mathbb{A}(\genGam,1/d_{\max},K)}\frac{K\alNt^{\frac{K-1}{2}}2^{-\frac{K-1}{K}(R_{\rm BT}+\Delta R)}}
{\frac{\aNorm^{K-1}}{2^{R_{\rm BT}}}{\sqrt{d_{\max}}}}.
\end{align}

\section{Proof of Theorem~\ref{thm:thm1}}
\label{app:proofOfThm1}
To establish Theorem~\ref{thm:thm1}, we follow the footsteps of the proof of Theorem~\ref{thm:thm1} in \cite{OutageBehaviorOfRandomlyPrecoded:Domanovitz2015} to obtain
\begin{align}
&\sum_{{\bf a}\in{\mathbb{A}(\genGam,1/d_{\max},K)}}\prob\left(\|\svv{D}^{-1/2}\svv{U}^T{\bf a}\|<\sqrt{\genGam}\right) \nonumber \\
&\leq
\sum_{{\bf a}\in{\mathbb{A}(\genGam,1/d_{\max},K)}}\frac{K{\sqrt{\genGam}}^{K-1}}{\frac{\aNorm^{K-1}}{2^{R_{\rm BT}}}{\sqrt{d_{\max}}}}
\label{eq:summation}
\end{align}
where $\mathbb{A}(\genGam,1/d_{\max},K)$ and $\genGam$ are defined in (\ref{eq:A_beta_d_general}) and (\ref{eq:beta}), respectively. Noting that
$$\mathbb{A}(\genGam,1/d_{\max},K) \subseteq \left\{\svv{a}: \|\svv{a}\| \leq \left\lfloor\sqrt{{\genGam}{d_{\max}}}\right\rfloor+1 \right\},$$
the summation in \eqref{eq:summation} can be bounded as
\begin{align}
\leq\sum^{\lfloor\sqrt{\genGam d_{\max}}\rfloor}_{k=0}\sum_{k<\|{\bf a}\|\leq k+1}\frac{K{\sqrt{\genGam}}^{K-1}}{\frac{k^{K-1}}{2^{R_{\rm BT}}}{\sqrt{d_{\max}}}}
\label{eq:96}
\end{align}
We apply Lemma 1 in \cite{OrdentlichErez2012} (a bound for the number of integer vectors contained in a ball of a given radius). Using this bound while noting that when $\|{\bf a}\|=1$ there are exactly $K$ integer vectors, the right hand side of (\ref{eq:96}) may be further bounded as
\begin{align}
&\leq \frac{K{\sqrt{\genGam}}^{K-1}2^{R_{\rm BT}}}{{\sqrt{d_{\max}}}}{\rm Vol}(\mathcal{B}_{K}(1)) \times \left[K + \sum^{\lfloor\sqrt{\genGam d_{\max}}\rfloor}_{k=1} \right. \nonumber \\
& \left. ~~~~~~~\frac{\left(k+1+\frac{\sqrt{K}}{2}\right)^{K}-\left(\max\left(k-\frac{\sqrt{K}}{2},0\right)\right)^{K}}{k^{K-1}}\right]
\label{eq:81}
\end{align}
where we note that  (\ref{eq:81}) trivially holds when
$\left\lfloor\sqrt{{\genGam}{d_{\max}}}\right\rfloor=0$
since $\mathbb{A}(\genGam,1/d_{\max},K)$ is the empty set in this case
and hence the left hand side of \eqref{eq:ratioForSUC}
evaluates to zero.

The right hand side of (\ref{eq:81}) can further be rewritten as
\begin{align}
& \frac{K{\sqrt{\genGam}}^{K-1}2^{R_{\rm BT}}}{{\sqrt{d_{\max}}}}{\rm Vol}(\mathcal{B}_{K}(1)) \times \left[ \underbrace{K}_{\rm I} + \underbrace{\sum^{\left\lfloor\frac{\sqrt{K}}{2}\right\rfloor}_{k=1}\frac{\left(k+1+\frac{\sqrt{K}}{2}\right)^{K}}{k^{K-1}}}_{\rm II} + \right. \nonumber \\
& \left. \underbrace{\sum^{\left\lfloor\sqrt{{\genGam d_{\max}}}\right\rfloor}_{k=\left\lfloor\frac{\sqrt{K}}{2}\right\rfloor+1}\frac{\left[\left(k+1+\frac{\sqrt{K}}{2}\right)^{K}-\left(k-\frac{\sqrt{K}}{2}\right)^{K}\right]}{k^{K-1}}}_{\rm III}
\right].
\label{eq:82}
\end{align}
We search for $c_1$ and $c_2$ (independent of $k$) such that
\begin{align}
\left(k+1+\frac{\sqrt{K}}{2}\right)^{K}\leq c_1 k^{K-1}
\label{eq:c1}
\end{align}
for $1\leq k \leq \left\lfloor\frac{\sqrt{K}}{2}\right\rfloor$, and
\begin{align}
\left[\left(k+1+\frac{\sqrt{K}}{2}\right)^{K}-\left(k-\frac{\sqrt{K}}{2}\right)^{K}\right]\leq c_2 k^{K-1}
\label{eq:c2}
\end{align}
for $k\geq1$, since it will then follow that
\begin{align}
II+III &\leq \sum^{\left\lfloor\sqrt{{\genGam d_{\max}}}\right\rfloor}_{k=1} \max(c_1,c_2) \nonumber \\
&={\left\lfloor\sqrt{{\genGam d_{\max}}}\right\rfloor} \max(c_1,c_2).
\end{align}
We note that since (again assuming $\left\lfloor\sqrt{{\genGam}{d_{\max}}}\right\rfloor\geq1$)
\begin{align}
K \leq \sum^{\left\lfloor\sqrt{{\genGam}{d_{\max}}}\right\rfloor}_{k=1} K,
\end{align}
it will thus follow that
\begin{align}
I+II+III &\leq  \sum^{\left\lfloor\sqrt{{\genGam}{d_{\max}}}\right\rfloor}_{k=1}  \left[K+\max(c_1,c_2)\right] \nonumber \\
&={\left\lfloor\sqrt{{\genGam}{d_{\max}}}\right\rfloor} \left[K+\max(c_1,c_2)\right].
\label{eq:102}
\end{align}

An explicit derivation for $c_1$ and $c_2$ appears in Appendix B of \cite{OutageBehaviorOfRandomlyPrecoded:Domanovitz2015} (where $2\Nt$ should be replaced with $K$), from which we obtain
\begin{align}
c_1&=\left(1+\sqrt{K}\right)^{K} \nonumber \\
c_2&=\left[\left(2+\frac{\sqrt{K}}{2}\right)^{K}-\left(1-\frac{\sqrt{K}}{2} \right)^{K}\right].
\end{align}
In is also shown in Appendix B of \cite{OutageBehaviorOfRandomlyPrecoded:Domanovitz2015}  that for $K\geq 4$, $c_2\leq c_1$ holds. For $1\leq K < 4$ we observe that  $c_1\leq c_2$. This is so since  $K<4$ implies that $1-\frac{\sqrt{K}}{2}>0$, and hence  indeed for $K<4$ we have
\begin{align}
    c_2 &= \left[\left(2+\frac{\sqrt{K}}{2}\right)^{K}-\left(1-\frac{\sqrt{K}}{2} \right)^{K}\right] \nonumber \\
    & = \left(1+ \sqrt{K}+1-\frac{\sqrt{K}}{2}\right)^{K}-\left(1-\frac{\sqrt{K}}{2} \right)^{K} \nonumber \\
    & \geq (1+\sqrt{K})^K \nonumber \\
    & = c_1.
\end{align}
Recalling now (\ref{eq:102}) and denoting
\begin{align}
    c_{\max}=
    \begin{cases}
    c_2 ~~~ 1\leq K < 4 \\
    c_1 ~~~ K\geq 4
    \end{cases},
\end{align}
it follows that
\begin{align}
    I+II+III &\leq {\left\lfloor\sqrt{{\genGam d_{\max}}}\right\rfloor} (K+c_{\max}).
    \label{eq:I+II}
\end{align}

Using \eqref{eq:I+II} and substituting the  volume of a unit ball ${\rm Vol}(\mathcal{B}_{K}(1))=\frac{\pi^{K/2}}{\Gamma(K/2+1)}$,  it follows that right hand side of (\ref{eq:81}) is  upper bounded by
\begin{align}
&\frac{K{\sqrt{\genGam}}^{K-1}2^{R_{\rm BT}}}{{\sqrt{d_{\max}}}} \frac{\pi^{K/2}}{\Gamma(K/2+1)} \sqrt{\genGam d_{\max}} (K+c_{\max}) \\
&=K{\sqrt{\genGam}}^{K}2^{R_{\rm BT}}(K+c_{\max})\frac{\pi^{K/2}}{\Gamma(K/2+1)}.
\label{eq:thm2NoD}
\end{align}
Finally, we substitute  $\genGam$, as defined in (\ref{eq:beta}), into \eqref{eq:thm2NoD} to obtain
\begin{align}
    &\sum_{{\bf a}\in{\mathbb{A}(\genGam,1/d_{\max},K)}}\prob\left(\|\svv{D}^{-1/2}\svv{U}^T{\bf a}\|<\sqrt{\genGam}\right)  \nonumber \\
    &\leq K\left(\alNt 2^{-\frac{2}{K}(R_{\rm BT}+\Delta R)}\right)^{K/2}2^{R_{\rm BT}}(K+c_{\max}) \nonumber \\
    &\leq K\left(\alNt \right)^{K/2} 2^{-(R_{\rm BT}+\Delta R)} 2^{R_{\rm BT}}(K+c_{\max}) \nonumber \\
    &\leq K\left(\alNt \right)^{K/2} (K+c_{\max}) 2^{-\Delta R} \nonumber \\
    &= c(K) 2^{-\Delta R},
\end{align}
where $c(K)$ is as defined in \eqref{eq:ckdef}.

\section{Proof of Lemma~\ref{lem:lem2} and Theorem \ref{thm:thm2}}
\label{app:proofOfThm2}
We first recall a theorem of Minkowski \cite[Theorem 1.5]{micciancio2012complexity}
that upper bounds the
product of the successive minima.
\begin{theorem}[Minkowski]
For any lattice $\Lambda(\svv{F}^T)$ that is spanned by a full rank $K\times K$ matrix $\svv{F}^T$
\begin{align}
    \prod_{m=1}^{K}\lambda_m^2(\svv{F}^T)\leq K^{K}(\det{\svv{F}^T})^2.
\end{align}
\label{thm:Minkowski}
\end{theorem}
To prove Theorem \ref{thm:thm2}, we further need the following two lemmas.
\begin{lemma} For a Gaussian source vector with covariance matrix $\svv{K_{xx}}\in\mathcal{K}(R_{\rm BT})$,
and for any full-rank integer matrix $\svv{A}$, the sum-rate of IF-SUC satisfies
\begin{align}
\sum_{m=1}^K R_{{\rm IF-SUC},m}(\svv{K}_{\bf x x};\svv{A})=R_{\rm BT}+\log|\det(A)|,
\end{align}
where $R_{{\rm IF-SUC},m}(\svv{K}_{\bf x x};\svv{A})$ is defined in \eqref{eq:R_IFSUC_m}.
\label{lem:appLem1}
\end{lemma}
\begin{proof}
\begin{align}
\sum_{m=1}^K R_{{\rm IF-SUC},m}(\svv{K}_{\bf x x};\svv{A})&=\frac{1}{2} \sum_{m=1}^K \log(\ell_{m,m}^2) \nonumber \\
& = \frac{1}{2}\log\left(\prod_{m=1}^K \ell_{m,m}^2 \right) \nonumber \\
& = \frac{1}{2}\log \det (\svv{L}\svv{L}^T) \nonumber \\
& = \frac{1}{2}\log \det \left(\svv{A}\left(\svv{I}+\frac{1}{d}\svv{K}_{\bf xx}\right)\svv{A}^T\right) \nonumber \\
& = R_{\rm BT} +\log|\det(A)|
\end{align}
\end{proof}

Theorem~3 in \cite{OrdentlichErezNazer:2013} shows that for successive IF (used for channel coding) there is no
loss (in terms of achievable rate) in restricting $\svv{A}$ to the class of unimodular matrices. We note that same claim  holds also in our framework (that of successive integer-source coding) by replacing $\svv{G}$, the matrix  spanning the lattice which was defined in Theorem 3 as
\begin{align}
\left(\svv{I}+\SNR\svv{H}\svv{H}^T\right)^{-1}=\svv{G}\svv{G}^T
\end{align}
with $\svv{F}$, as defined in (\ref{eq:def_F}), and noting that the optimal $\svv{A}$ can be expressed (in both cases) as
\begin{align}
\svv{A}^{\rm opt}_{\rm SIC/SUC}=\min_{\substack{\svv{A}\in\mathbb{Z}^{K\times K} \\ \det{\svv{A}}\neq0}}\max_{k=1,...,K} \ell_{k,k}^2,
\end{align}
where $\ell_{k,k}$ are the  diagonal elements of the corresponding $\svv{L}$ matrix (derived from the Cholesky decomposition) in each case.

Having established that the optimal $\svv{A}$ is unimodular, it now follows  that $\sum_{m=1}^K R_{{\rm IF-SUC},m}(\svv{K}_{\bf x x};\svv{A}^{\rm opt}_{\rm SUC}) = R_{\rm BT}$.


We are now ready to prove the following lemma that is analogous to Theorem 3 in \cite{ordentlich2014approximate}.
First, denote by
\begin{align}
    R_{{\rm IF},k}(\svv{K}_{\bf x x};\svv{A})\triangleq \frac{1}{2}\log\left({\bf a}_k^T\left(\svv{I}+\svv{K}_{\bf x x}\right){\bf a}_k\right)
    \label{eq:comprate_m}
\end{align}

the effective rate that of the $k$'th equation as appears in the definition of the achievable rate of integer-forcing source coding in \eqref{eq:R_IF(A)}. Then, we have the following.
\begin{lemma}
For a Gaussian source vector with covariance matrix $\svv{K_{xx}}\in\mathcal{K}(R_{\rm BT})$,
and for the optimal integer matrix $\svv{A}^{\rm opt}$, the sum-rate of IF is upper bounded as
\begin{align}
\sum_{m=1}^K R_{{\rm IF},m}(\svv{K}_{\bf x x};\svv{A}^{\rm opt}) &\leq R_{\rm BT}+\frac{K}{2}\log(K).
\end{align}
\label{lem:appLem2}
\end{lemma}

\begin{proof}
\begin{align}
\sum_{m=1}^K R_{{\rm IF},m}(\svv{K}_{\bf x x};\svv{A}^{\rm opt})&=\sum_{m=1}^K \frac{1}{2} \log\left(\lambda_m^2(\svv{F}^T)\right) \nonumber \\
& = \frac{1}{2} \log\left( \prod_{m=1}^K \lambda_m^2(\svv{F}^T)\right) \nonumber \\
& \leq \frac{1}{2} \log\left(  K^K |\det(F)|^2 \right) \nonumber \\
& = R_{\rm BT} + \frac{K}{2} \log\left(K\right).
\end{align}
where the inequality is due to  Theorem~\ref{thm:Minkowski} (Minkowski's Theorem).
\end{proof}

Now, for the case of two sources we have by Lemma~\ref{lem:appLem1}
\begin{align}
R_{{\rm IF-SUC},1}(\svv{K}_{\bf x x};\svv{A}^{\rm opt}_{\rm SUC})+R_{{\rm IF-SUC},2}(\svv{K}_{\bf x x};\svv{A}^{\rm opt}_{\rm SUC})=R_{\rm BT},
\label{eq:108}
\end{align}
or equivalently
\begin{align}
R_{{\rm IF-SUC},2}(\svv{K}_{\bf x x};\svv{A}^{\rm opt}_{\rm SUC})=R_{\rm BT}-R_{{\rm IF-SUC},1}(\svv{K}_{\bf x x};\svv{A}^{\rm opt}_{\rm SUC}).
\label{eq:109}
\end{align}
We further have by lemma~\ref{lem:appLem2} that
\begin{align}
R_{{\rm IF},1}(\svv{K}_{\bf x x};\svv{A}^{\rm opt})+R_{{\rm IF},2}(\svv{K}_{\bf x x};\svv{A}^{\rm opt})\leq R_{\rm BT}+1.
\label{eq:110}
\end{align}
 We note that the (optimal) integer matrix $\svv{A}^{\rm opt}$ used for IF
in \eqref{eq:110} is in general different than the (optimal) matrix $\svv{A}^{\rm opt}_{\rm SUC}$ used for IF-SUC in \eqref{eq:108}-\eqref{eq:109}.
Nonetheless, when applying IF-SUC, one decodes first the equation with the \emph{lowest} rate. Since for this equation SUC has no effect, it follows that the first row of $\svv{A}$ is the same in both cases and hence
\begin{align}
R_{{\rm IF-SUC},1}(\svv{K}_{\bf x x};\svv{A}^{\rm opt}_{\rm SUC})=R_{{\rm IF},1}(\svv{K}_{\bf x x};\svv{A}^{\rm opt}).
\end{align}
Since source $1$ is decoded first, it follows that $R_{{\rm IF},1}>R_{{\rm IF},2}$ and hence
\begin{align}
R_{{\rm IF},1}(\svv{K}_{\bf x x};\svv{A}^{\rm opt})\leq\frac{R_{\rm BT}+1}{2}.
\end{align}
Therefore,
\begin{align}
R_{\rm IF-SUC}^{\rm opt}(\svv{K}_{\bf x x})&=2\max \left(R_{{\rm IF-SUC},1}(\svv{K}_{\bf x x};\svv{A}^{\rm opt}_{\rm SUC}),R_{{\rm IF-SUC},2}(\svv{K}_{\bf x x};\svv{A}^{\rm opt}_{\rm SUC}) \right) \nonumber \\
& \leq 2 \max \left(\frac{R_{\rm BT}+1}{2},R_{\rm BT}-R_{{\rm IF},1}(\svv{K}_{\bf x x};\svv{A}^{\rm opt})\right) \nonumber \\
& = \max \left(R_{\rm BT}+1,2R_{\rm BT}-2R_{{\rm IF},1}(\svv{K}_{\bf x x};\svv{A}^{\rm opt})\right).
\end{align}
Henceforth, we analyze the outage for $R_{\rm BT}>1$ and target rates that are no smaller than $R_{\rm BT}+1$, so that the inequality $2R_{\rm BT}-2R_{{\rm IF},1}(\svv{K}_{\bf x x})>R_{\rm BT}+1$ is satisfied. Thus, we consider excess rate values satisfying $\Delta R>1$. Our goal is to bound
\begin{align}
&\prob \left(R_{\rm IF-SUC}^{\rm opt}(\svv{K}_{\bf x x})\geq R_{\rm BT}+\Delta R \right)  \nonumber \\
&=\prob \left(2R_{\rm BT}-2R_{{\rm IF},1}(\svv{K}_{\bf x x};\svv{A}^{\rm opt}_{\rm SUC})\geq R_{\rm BT}+\Delta R \right) \nonumber \\
&=\prob \left(R_{{\rm IF},1}(\svv{K}_{\bf x x};\svv{A}^{\rm opt})\leq \frac{R_{\rm BT}-\Delta R}{2} \right) \nonumber \\
&=\prob \left(\frac{1}{2}\log(\lambda_1^2(\svv{F}^T)) \leq \frac{R_{\rm BT}-\Delta R}{2} \right) \nonumber \\
&=\prob \left(\lambda_1^2(\svv{F}^T)\leq 2^{R_{\rm BT}-\Delta R} \right)
\label{eq:PrLambda1Implicit-Comp_SUC}
\end{align}

Let $\genGam=2^{R_{\rm BT}-\Delta R}$. We wish to bound~(\ref{eq:PrLambda1Implicit-Comp_SUC}), or equivalently
\begin{align}
\prob \left(\lambda_1^2(\svv{F}^T)\leq \genGam \right)=\prob \left(\lambda_1(\svv{F}^T)\leq \sqrt{\genGam} \right)
\end{align}
for a given matrix $\svv{D}$ corresponding to $\svv{K}_{\bf x x}$ (via the relation \eqref{eq:KDconnection}). Note that the event $\lambda_1(\Lambda)<\sqrt{\genGam}$ is equivalent to the event
\begin{align}
\bigcup_{\bf{a}\in{\mathbb{Z}^{K}\setminus\{{\bf 0}\}}}||\svv{D}^{1/2}\svv{U}^T\bf{a}||<\sqrt{\genGam}.
\end{align}
Applying the union bound yields
\begin{align}
\prob\left(\lambda_1(\Lambda)<\sqrt{\genGam}\right) \leq  \sum_{{\bf a}\in {\mathbb{Z}^{K} \setminus\{{\bf 0}\} }}\prob\left(||\svv{D}^{1/2}\svv{U}^T{\bf a}||<\sqrt{\genGam}\right).
\end{align}
Using the same derivation as in Lemma~2 in \cite{OutageBehaviorOfRandomlyPrecoded:Domanovitz2015}, we get
\begin{align}
&\prob\left(R_{\rm IF-SUC}^{\rm opt}(\svv{D},\svv{U})\geq R_{\rm BT}+\Delta R\right)
\nonumber \\
& \leq  \sum_{{\bf a}\in\mathbb{A}(\genGam,d_{\min},K)}\frac{K\sqrt{\genGam}^{K-1}}{\|{\bf a}\|^{K-1}2^{R_{\rm BT}}\frac{1}{\sqrt{d_{\min}}}}.
\end{align}
Since we are analyzing the case of two sources, we have
\begin{align}
& \prob\left(R^{\rm opt}_{\rm  IF-SUC}(\svv{D},\svv{U})\geq R_{\rm BT}+\Delta R\right)
 \nonumber \\
& \leq \sum_{{\bf a}\in\mathbb{A}(\genGam,d_{\min},K)}\frac{2\sqrt{\genGam}}{\|{\bf a}\|2^{R_{\rm BT}}\frac{1}{\sqrt{d_{\min}}}}.
\end{align}
This establishes Lemma~\ref{lem:lem2}.

Applying a similar argument as appears in Appendix~\ref{app:proofOfThm1} (as part of the proof of Theorem~\ref{thm:thm1}), and noting that $K<4$ implies that $c_{\max}=c_2$, we get
\begin{align}
\prob\left(R^{\rm opt}_{\rm  IF-SUC}(\svv{D},\svv{U})\geq R_{\rm BT}+\Delta R)\right) \nonumber \\
\leq  \frac{2 \genGam}{2^{R_{\rm BT}}}\left(2+3+3\sqrt{2}\right)\pi.
\end{align}

Finally, substituting $\genGam$ as defined in \eqref{eq:beta}, we obtain
\begin{align}
& \prob\left(R^{\rm opt}_{\rm  IF-SUC}(\svv{D},\svv{U})\geq R_{\rm BT}+\Delta R\right)
 \nonumber \\
& \leq \frac{2\pi\cdot2^{R_{\rm BT}-\Delta R}}{2^{R_{\rm BT}}}\left(2+3+3\sqrt{2}\right) \nonumber \\
& =2\pi \left(5+3\sqrt{2}\right) 2^{-\Delta R}.
\end{align}

\ifCLASSOPTIONcaptionsoff
  \newpage
\fi

\section{Additive (Worst-Case) Bound for NVD Space-Time Precoded  Sources}
\label{app:additiveBound}

Combining space-time precoding and integer forcing in the context of channel coding was suggested in \cite{Domanovitz:2014}, as we next briefly recall. We then present the necessary modifications for the case of source coding.

We derive below an additive bound using a unitary precoding matrix satisfying a non-vanishing determinant (NVD) property. As the theory of NVD space-time codes has been developed over the complex field, it
will prove convenient for us to employ complex precoding matrices. To this end, we may assume that we stack samples from two time slots where the samples stacked at the first time slot represent the real part of a complex number and the samples stacked at the second time slot represent the imaginary part of a complex number. Hence, we have
\begin{align}
\svv{K}_{\hat{x}\hat{x}}&=\svv{I}_{2\times 2} \otimes \svv{K}_{x x},
\end{align}
where $\otimes$ is the Kronecker product. We note that the Berger-Tung benchmark of this stacked source vector is
\begin{align}
R_{\rm BT}(\svv{K}_{\hat{x}\hat{x}})=2R_{\rm BT}(\svv{K}_{{x}{x}}).
\label{eq:R_btComplex}
\end{align}

Next, in order to allow space-time precoding, we stack $T$ times the $K$ ``complex'' outputs of the $K$ sources and let $\bar{\boldsymbol{x}}_c\in \mathbb{R}^{2 K T \times 1}$ denote the effective source vector. Its correlation matrix, which  we refer to as the effective covariance matrix, takes the form
\begin{align}
    \mathcal{K}&=\svv{I}_{T\times T} \otimes \svv{K}_{\hat{x} \hat{x}}.
    \label{eq:stackedSource}
\end{align}

We assume a precoding matrix,  that in principle can be either deterministic or random, is applied to the effective source vector.
We analyze performance for the case where the precoding matrix $\svv{P}_{st,c}$ is deterministic, specifically a precoding matrix induced by a perfect space-time block code, operating on the stacked source $\bar{\boldsymbol{x}}_c$ having covariance matrix as given in (\ref{eq:stackedSource}). An explanation on how to extract the precoding matrix from a space-time code can be found in Section IV in \cite{ordentlich2015precoded}.

We denote the corresponding precoding matrix over the reals as $\svv{P}_{st} \in \mathbb{R}^{2K T \times 2K T}$. We denote
\begin{align}
    \svv{I}+\svv{P}_{st}\mathcal{K}\svv{P}_{st}^H=\mathcal{F}\mathcal{F}^T.
\end{align}

As we assume that the precoding matrix is unitary, the Berger-Tung benchmark (normalized by the total number of time extensions used) remains unchanged, i.e.,
\begin{align}
    \frac{1}{2T}\frac{1}{2}\log\det\left(\svv{I}+\svv{P}_{st}\mathcal{K}\svv{P}_{st}^H\right)&=\frac{1}{2T}\frac{1}{2}\log\left(\det\mathcal{F}^T\right)^2 \nonumber \\
    &=R_{\rm BT}.
\end{align}

As noted above, we assume that the generating matrix of a perfect code \cite{oggier2006perfect,ExplicitSpaceTimeCodesAchievingtheDiversity:Elia2006} is employed as a precoding matrix.

A space-time code is called perfect if:
\begin{itemize}
\item
It is full rate;
\item
It satisfies the non-vanishing determinant (NVD) condition;
\item
The code's generating matrix is unitary.
\end{itemize}

Let $\delta_{\min}$ denote the minimal non-vanishing determinant of this code.
Such codes further use the minimal number of time extensions possible, i.e., $T=K$.
Thus, we have a total of $K^2$ stacked complex samples. Subsisting $2K^2$ as the dimension (number of real samples jointly processed) in  (\ref{eq:RSumif}), the rate of IF source coding for
the time-stacked samples is given by
\begin{align}
R_{\rm IF}^{\rm opt}(\mathcal{K},\svv{P}_{st})=(2K^2)\frac{1}{2}\log\left(\lambda_{2K^2}^2(\mathcal{F}^T)\right).
\end{align}

To bound $\lambda_{2K^2}^2(\mathcal{F}^T)$, we note that for every $2K^2$-dimensional lattice, we have
\begin{align}
\lambda_1^{2(2K^2-1)}(\mathcal{F}^T)\lambda_{2K^2}^2(\mathcal{F}^T)\leq \prod_{m=1}^{2K^2}\lambda_{m}^2(\mathcal{F}^T).
\end{align}
Using Minkowski's theorem (Theorem \ref{thm:Minkowski}
 in Appendix~\ref{app:proofOfThm2}), it follows that
 \begin{align}
 \lambda_{2K^2}^2(\mathcal{F}^T)\leq (2K^2)^{2K^2}(\det{\mathcal{F}^T})^2 \frac{1}{\lambda_{1}^{2(2K^2-1)}(\mathcal{F}^T)}.
 \end{align}

 Hence, the rate of IF source coding (normalized by the number of time extensions) can be bounded as:
{\small
\begin{align}
    &\frac{1}{2K}R_{\rm IF}^{\rm opt}(\mathcal{K},\svv{P}_{st}) \nonumber \\
    &= \frac{1}{2K}(2K^2)\frac{1}{2}\log\left(\lambda_{2K^2}^2(\mathcal{F}^T)\right) \nonumber \\
    & \leq\frac{K}{2} \log\left((2K^2)^{2K^2}(\det{\mathcal{F}^T})^2 \frac{1}{\lambda_{1}^{2(2K^2-1)}(\mathcal{F}^T)}\right)  \nonumber \\
    & = K^3\log(2K^2)+\frac{K}{2}\log(\det({\mathcal{F}^T})^2)-\frac{K(2K^2-1)}{2}\log(\lambda_{1}^{2}(\mathcal{F}^T)) \nonumber \\
    & = K^3\log(2K^2)+2K^2R_{\rm BT} -\frac{K(2K^2-1)}{2}\log(\lambda_{1}^{2}(\mathcal{F}^T)).
    \label{eq:SubMink}
\end{align}
}%

We next use the results derived in  \cite{ordentlich2015precoded} for channel coding (using NVD precoding). We note that since the covariance matrix is positive semi-definite, it may be written as
\begin{align}
\svv{K}_{{x}{x}}=\svv{H}\svv{H}^T.
\end{align}
The covariance matrix of the stacked source vector may be written as
\begin{align}
\svv{K}_{\hat{x}\hat{x}}=\svv{\hat{H}}\svv{\hat{H}}^T,
\end{align}
where we may take $\svv{\hat{H}}=\svv{I}_{2\times 2}\otimes\svv{H}$.

There are many such choices of $\svv{H}$ and any such choice corresponds to a channel matrix $\svv{\hat{H}}$ that can be viewed as the real representation of a complex channel matrix (which in the present case is real, i.e., has no imaginary part) in the context of \cite{ordentlich2015precoded}.
The effective covariance matrix can similarly be rewritten as
\begin{align}
\mathcal{K}=\mathcal{H}\mathcal{H}^T,
\end{align}
where $\mathcal{H}=\svv{I}_{T\times T} \otimes
\svv{\hat{H}}$.

Using the channel coding terminology
of \cite{ordentlich2015precoded}, we further define the minimum distance at the receiver $d_{\rm min}(\svv{H},L)$ as
\begin{align}
d_{\rm min}(\svv{H},L)\triangleq \min_{\bf{a}\in {\rm QAM}^{K}(L) \backslash 0} \|\svv{H}{\bf a}\|,
\end{align}
where
\begin{align}
{\rm QAM}(L)\triangleq & \{-L,-L+1,\ldots,L-1,L\} \nonumber \\
                    & +i\{-L,-L+1,\ldots,L-1,L\}.
\end{align}

Setting $\SNR=1$ and for $\svv{\hat{H}}\in\mathbb{R}^{2K\times2K}$ (which is the real representation of $\svv{H}_c\in\mathbb{C}^{K\times K}$), Lemma 2 in \cite{ordentlich2015precoded} states that
\begin{align}
\frac{1}{4K^4}\min_{{\bf a}\in\mathbb{Z}^{2K}\setminus 0}{\bf a}^T\left(\svv{I}+\svv{P}_{st}^H\mathcal{H}^H\mathcal{H}\svv{P}_{st}\right){\bf a}
  \nonumber \\
 \geq \frac{1}{4K^4}\min_{L=1,2,\cdots} \left(L^2+\SNR d_{\min}(\mathcal{H}\svv{P}_{st},L)\right).
 \end{align}
 Using Corollary~1 in \cite{ordentlich2015precoded}, we get
\begin{align}
 &\frac{1}{4K^4}\min_{{\bf a}\in\mathbb{Z}^{2K}\setminus 0}{\bf a}^T\left(\svv{I}+\svv{P}_{st}^H\mathcal{H}^H\mathcal{H}\svv{P}_{st}\right){\bf a}
  \nonumber \\
 & \geq\frac{1}{4K^4}\min_{L=1,2,\cdots} \left(L^2+\left[\delta_{\min}^\frac{1}{K}2^\frac{C_{\rm WI}}{K}-2K^2L^2\right]^{+}\right) \nonumber \\
 & \geq \frac{1}{4K^4}\min_{L=1,2,\cdots} \left(L^2+\left[\frac{\delta_{\min}^\frac{1}{K}2^\frac{C_{\rm WI}}{K}}{2K^2}-L^2\right]^{+}\right) \nonumber \\
 & \geq \frac{1}{8K^6} \delta_{\min}^\frac{1}{K}2^\frac{C_{\rm WI}}{K},
\end{align}
where $C_{\rm WI}=\frac{1}{2}\log\det\left(\svv{I}+\svv{H}^H\svv{H}\right)$ is the mutual information of $\svv{H}$. Since $C_{\rm WI}$ is the rate of a $2K\times 2K$ real matrix (resulting from a $K\times K$ complex matrix), it equals $R_{\rm BT}(\svv{K}_{\hat{x}\hat{x}})$ defined in (\ref{eq:R_btComplex}). Hence, we obtain
\begin{align}
\min_{{\bf a}\in\mathbb{Z}^{2K}\setminus 0}{\bf a}^T\left(\svv{I}+\svv{P}_{st}^H\mathcal{H}^H\mathcal{H}\svv{P}_{st}^H\right){\bf a}\geq\frac{1}{2K^2}\delta_{\min}^{\frac{1}{K}}2^{\frac{2R_{\rm BT}}{K}}
\end{align}
which in turn yields
\begin{align}
\lambda_{1}^{2}(\mathcal{F}^T)&=\min_{{\bf a}\in\mathbb{Z}^{2K}\setminus 0}{\bf a}^T\left(\svv{I}+\svv{P}_{st}\mathcal{H}\mathcal{H}^T\svv{P}_{st}^T\right){\bf a} \nonumber \\
&\geq \frac{1}{2K^2}\delta_{\min}^{\frac{1}{K}}2^{\frac{2R_{\rm BT}}{K}}.
\label{eq:boundForShortest}
\end{align}

Finally, plugging the bound (\ref{eq:boundForShortest}) into (\ref{eq:SubMink}), we arrive at
\begin{align}
&\frac{1}{2K}R_{\rm IF}^{\rm opt}(\mathcal{K},\svv{P}_{st})  \nonumber \\
& \leq 2K^2 R_{\rm BT}+K^3\log(2K^2)+\frac{K(2K^2-1)}{2}\log(2K^2) \nonumber \\
&+\frac{K(2K^2-1)}{2K}\log(\frac{1}{\delta_{\min}})-\frac{K(2K^2-1)}{2K} 2R_{\rm BT}
\nonumber \\
& \leq R_{\rm BT}+2K^3\log(2K^2)+K^2\log\frac{1}{\delta_{\min}}.
\end{align}

\section{Proof of Lemma~\ref{lem:lem3}}
\label{app:proofOfLemma3}
A Gaussian source component with a specific rate $R_i$ can be transformed to a different Gaussian source with rate $R_i^{\Delta}$ by appropriate scaling.
Specifically, scaling each source component $i$ by
\begin{align}
    \alpha_i=\sqrt{\frac{2^{2R_i^{\Delta}}-1}{2^{2R_i}-1}}
    \label{eq:scaleFactor}
\end{align}
results in parallel (uncorrelated) sources
\begin{align}
\hat{{x}_i}=\alpha_i{x}_i
\end{align}
with variances
\begin{align}
\hat{{s}}_i^2=\alpha_i^2{s}_i^2.
\end{align}
By (\ref{eq:R_S}), we therefore have
\begin{align}
    \hat{R}_i & =\frac{1}{2}\log\left(1+\hat{s}_i^2\right) \nonumber \\
              & =\frac{1}{2}\log\left(1+\alpha_i^2 s_i^2\right) \nonumber \\
              & =\frac{1}{2}\log\left(1+\frac{2^{2R_i^{\Delta}}-1}{2^{2R_i}-1} s_i^2\right) \nonumber \\
              & =\frac{1}{2}\log\left(1+\frac{2^{2R_i^{\Delta}}-1}{2^{2R_i}-1} (2^{2R_i}-1)\right) \nonumber \\
              & =\frac{1}{2}\log\left(1+2^{2R_i^{\Delta}}-1\right) \nonumber \\
              & = R_i^{\Delta}.
\end{align}

We associate with any rate tuple $(R_1, R_2, \ldots, R_M)\in \mathcal{R}(R_{\rm BT})$ a   rate tuple $(R_1^{\Delta}, R_2^{\Delta}, \ldots, R_M^{\Delta}) \in \mathcal{R}^{\Delta}(R_{\rm BT})$, according to  the following transformation
\begin{align}
    R_i^{\Delta}&=\left\lceil\frac{R_i}{\Delta R_{\rm BT}}\right\rceil\cdot \Delta R_{\rm BT} , \quad i=1,2,\ldots,K-1, \nonumber \\
    R_K^{\Delta}&=R_K-\sum_{i=1}^{K-1}(R_i^{\Delta}-R_i).
    \label{eq:rateScaleEq}
\end{align}

For $i=K$  we have
\begin{align}
    R_K^{\Delta} \geq R_K-(K-1)\Delta R_{\rm BT}.
\end{align}
It follows that the scaling factor needed to achieve $R_K^{\Delta}$ is bounded by
\begin{align}
\alpha_K^2&={\frac{2^{2R_K^{\Delta}}-1}{2^{2R_K}-1}}\nonumber \\
&\geq \frac{2^{2\left(R_K-(K-1)\Delta R_{\rm BT}\right)}-1}{2^{2R_K}-1} \nonumber \\
&\geq \frac{2^{2\left(\frac{R_{\rm BT}}{K}-(K-1)\Delta R_{\rm BT}\right)}-1}{2^{2\left(\frac{R_{\rm BT}}{K}\right)}-1}. \label{eq:66}
\end{align}
where (\ref{eq:66}) follows since it is readily verified that the
function $\frac{ax-1}{x-1}$ is monotonically increasing in $x$,
for $x\geq 1$ and $0 \leq a \leq 1$.
Denoting $\eta=\frac{1}{\alpha_k}$, it follows from (\ref{eq:66}) that
\begin{align}
    \eta^2=\frac{2^{2\left(\frac{R_{\rm BT}}{K}\right)}-1}{2^{2\left(\frac{R_{\rm BT}}{K}-(K-1)\Delta R_{\rm BT}\right)}-1}.
\end{align}

Now for $1\leq i \leq K-1$ we have by (\ref{eq:rateScaleEq}) that $R_i^{\Delta}\geq R_i$. Hence, for such $i$ it trivially holds that
(since $\eta \geq 1$)
\begin{align}
s_i^2 & \leq  s_{i,\Delta}^2 \\
      & \leq  \eta^2 s_{i,\Delta}^2.
\end{align}
Thus, the lemma follows by observing that from (\ref{eq:66}) it follows
that $s_K^2 \leq \eta^2 s_{K,\Delta}^2$ as well.

\section{Proof of Lemma~\ref{lem:lem4}}
\label{app:proofOfLemma4}

Recalling (\ref{eq:defOfv}), we note that (\ref{eq:quadratic}) can be written as
\begin{align}
    R_{{\rm IF},k}(\svv{S},\svv{P};\svv{A})=\frac{1}{2}\log\left(\sum_{i=1}^K {v}_i^2(1+s_i^2)\right).
\end{align}
Therefore, when scaling the  Gaussian input vector by a factor of $\beta\geq1$, we have
\begin{align}
    R_{{\rm IF},k}(\beta^2\svv{S},\svv{P};\svv{A})&=\frac{1}{2}\log\left(\sum_{i=1}^K { v}_i^2(1+\beta^2 s_i^2)\right)\nonumber \\
    &\leq \frac{1}{2}\log\left(\beta^2 \sum_{i=1}^K { v}_i^2(1+s_i^2)\right)\nonumber \\
    &=\frac{1}{2}\log\left(\beta^2\right)+ R_{{\rm IF},k}(\svv{S},\svv{P};\svv{A}) \nonumber \\
    &=\log\left(\beta\right)+ R_{{\rm IF},k}(\svv{S},\svv{P};\svv{A}).
    \label{eq:67}
\end{align}

\bibliographystyle{IEEEtran}
\bibliography{eladd}

\end{document}